\documentclass[11pt]{article}
\usepackage{lipsum}
\usepackage{amssymb, amsmath, amsthm, graphicx, authblk, bbm, fullpage, soul}
\usepackage{booktabs}
\usepackage{enumitem}

\usepackage[dvipsnames]{xcolor}

\usepackage{algorithm}% http://ctan.org/pkg/algorithms
\usepackage{algpseudocode}% http://ctan.org/pkg/algorithmicx
\algnewcommand\algorithmicinput{\textbf{INPUT:}}
\algnewcommand\INPUT{\item[\algorithmicinput]}
\algnewcommand\algorithmicoutput{\textbf{OUTPUT:}}
\algnewcommand\OUTPUT{\item[\algorithmicoutput]}

\usepackage{hyperref}[]
\hypersetup{
    colorlinks=true,
    linkcolor=blue,
    filecolor=magenta,      
    urlcolor=cyan,
      citecolor=blue,
    }
\usepackage{cleveref}

\providecommand{\keywords}[1]
{
  \small	
  \textbf{\textit{Keywords---}} #1
}

\newtheorem{theorem}{Theorem}
\newtheorem{lemma}[theorem]{Lemma}
\newtheorem{proposition}[theorem]{Proposition}
\newtheorem{corollary}[theorem]{Corollary}

\newtheorem{definition}{Definition}
\newtheorem{remark}{Remark}
\newtheorem{assumption}{Assumption}

\allowdisplaybreaks

\DeclareMathOperator*{\argmin}{arg\,min}

\DeclareMathOperator*{\op}{op}
\DeclareMathOperator*{\esssup}{ess\,sup}

\DeclareMathOperator*{\cov}{cov}

\usepackage[round]{natbib}

\title{Online network change point detection \\ with missing values and temporal dependence}

\author[1]{Haotian Xu}
\author[2]{Paromita Dubey}
\author[1]{Yi Yu}
\affil[1]{Department of Statistics, University of Warwick}
\affil[2]{Department of Data Sciences and Operations, Marshall School of Business, University of Southern California}
\date{\today}

\begin{document}

\maketitle

\begin{abstract}
 In this paper we study online change point detection in dynamic networks with time-heterogeneous missing pattern within networks and dependence across both nodes and time.  The missingness probabilities, the entrywise sparsity of networks, the rank of networks and the jump size in terms of the Frobenius norm, are all allowed to vary as functions of the pre-change sample size.  On top of a thorough handling of all the model parameters, we notably allow the edges and missingness to be dependent.  To the best of our knowledge, such general framework has not been rigorously nor systematically studied before in the literature. We propose a polynomial-time change point detection algorithm, with a version of soft-impute algorithm \citep[e.g.][]{mazumder2010spectral, klopp2015matrix} as the imputation sub-routine.  Piecing up these standard sub-routines algorithms, we are able to solve a brand new problem with sharp detection delay subject to an overall Type-I error control.  Extensive numerical experiments are conducted demonstrating the outstanding performances of our proposed method in practice.	
\end{abstract}

\keywords{Online change point, network, missing data, temporal dependence}

\section{Introduction}

In recent years, statistical analysis on dynamic networks proliferates, due to the soaring interest from application areas, including neuroscience \citep[e.g.][]{bassett2017network}, biology \citep[e.g.][]{kim2014inference}, climatology \citep[e.g.][]{cross201215}, finance \citep[e.g.][]{schuenemann2020japanese}, economics \citep[e.g.][]{wu2020dependency} and cybersecurity \citep[e.g.][]{basaras2013detecting}, to name a few.  To be specific, by dynamic networks, we mean a collection of relational observations (nodes) and their relation (edges), all of which may evolve along a certain linear ordering, referred to as time.  

Based on the form of available data and different sets of assumptions, there have been different types of research on dynamic networks.  When data are random observations associated with nodes, dynamic Gaussian graphical models and/or Markov random fields models are usually deployed \citep[see e.g.][]{keshavarz2018sequential, roy2017change}.  When observations are regarding edges, Erd\H{o}s--R\'enyi models, stochastic block models and random dot product models are often summoned \citep[see e.g.][]{chi2009evolutionary, zhou2010time, athreya2017statistical, arroyo2021inference}.  In terms of dynamic mechanisms, there are three main types of assumptions: the underlying distributions (1) stay unchanged during the time course \citep[e.g.][]{bhattacharyya2018spectral}, (2) slowly evolving along time allowing for small changes \citep[e.g.][]{zhou2010time, danaher2014joint}, and (3) abruptly change at certain unknown locations or otherwise unchanged \citep[e.g.][]{liu2018global, yu2021optimal}.

When assuming the situation (3) above, the problems are referred to as change point analysis.  Dynamic network change point detection problems have been studied from different angles in \cite{keshavarz2018sequential}, \cite{liu2018global}, \cite{cribben2017estimating}, \cite{wang2021optimal} and \cite{yu2021optimal}, among others.  Change point analysis on its own is an area with rich literature, and has been studied in univariate, multivariate, matrix, functional, manifold types of data, from both \emph{online} (detecting change points while collecting data) and \emph{offline} (estimating change points retrospectively with all data already available) perspectives, concerning testing and estimating goals.  We refer to \cite{yu2020review} and \cite{aminikhanghahi2017survey} for recent reviews.% on theoretical and computational change point analysis.  

Data collected in practice often contain missing values, especially for large-scale network-type data.  The statistical research on missing data enjoys a vast body of literature, including the renowned expectation-maximisation algorithm \citep[e.g.][]{dempster1977maximum, balakrishnan2017statistical}.  Due to the fact that the network data concerned in this paper are in the form of random matrices, we focus on the matrix completion literature, where \cite{candes2010matrix}, \cite{candes2010power}, \cite{candes2009exact} and others, are among the first line of attack introducing nuclear-norm penalisation estimators, based on low-rank assumptions.  Since then, different structural assumptions, missingness patterns and penalisation methods have been introduced \citep[e.g.][]{keshavan2010matrix, cai2016matrix, wang2015orthogonal}.  Theoretical guarantees of different estimators are also investigated, ranging from estimation error bounds to minimax optimality, for both one-step and iterative estimators \cite[e.g.][]{koltchinskii2011nuclear, klopp2014noisy, klopp2015matrix, klopp2011rank, bhaskar2016probabilistic, carpentier2018adaptive}.  

Missingness has also been previously considered in the change point analysis literature.  For high-dimensional vector-valued data with missing entries, \cite{xie2012change} devised a generalised likelihood ratio based statistic to study an online change point detection problem; \cite{foll:21} was concerned with the localisation error control in an offline perspective.  For Gaussian graphical models with potential missingness, \cite{londschien2021change} studied an offline change point detection and focused only on the computational aspects.  Other more application-oriented work includes \cite{muniz2011random} and \cite{yang2020changepoint}, among others.  \cite{enikeeva2021change}, arguably, is the closest-related to us in the literature.  \cite{enikeeva2021change} was concerned with testing the existing of and estimating the location of one change point in an offline setting, on dynamic networks with missingness.  Optimal offline detection and localisation rates were investigated in \cite{enikeeva2021change}. 

On top of the missingness, different from the aforementioned literature, we also consider temporal dependence in both (1) the network generating mechanism and (2) the missingness patterns, as well as the dependence between (1) and (2).  The generality and flexibility are the first of their kind.  Having said this, modelling only the dynamic network processes with temporal dependence has been considered in literature, such as Markov chain-based dynamic network models \citep[e.g.][]{snijders2005models, ludkin2018dynamic, jiang2020autoregressive}.  We focus on modelling dynamic networks using temporally dependent latent positions, under $\phi$-mixing conditions, with an example in \Cref{sec:example}.  The introduction of latent positions allows for edge-dependence within networks, in addition to temporal dependence.

In this paper, we assume that the data are a sequence of random dot product networks (\Cref{def-rdpg}), with a heterogeneous Bernoulli sampling missingness pattern (\Cref{def-sample}).  Both the dynamic networks and the missingness are driven by a sequence of temporally dependent latent positions associated with a set of fixed nodes.  The unknown underlying distribution of the networks is assumed to change at a certain time point, namely the change point.  Our goal is to detect this change point as soon as it occurs, while controlling the probability of false alarms. 

The contributions of this paper are summarised as follows.  Firstly, this is, to the best of our knowledge, the first paper studying online network change point detection possessing missing values with rigorous theoretical justifications.  In a rather general framework, we study a change point estimator based on the soft-impute matrix completion algorithm \citep[e.g.][]{mazumder2010spectral} and derive a nearly-optimal upper bound on the detection delay, while controlling the probability of false alarms.  Secondly, we allow for the $\phi$-mixing in the latent position sequence as well as the missing patterns.  Under such generality, we present a multiple-copy version of the soft-impute algorithm and verify the theoretical guarantees of its output, allowing for different copies possessing different missingness patterns.  Lastly, extensive numerical analysis is conducted to support our theoretical finding.

\textbf{Notation.} We do not distinguish a network and its adjacency matrix.  Let $\|\cdot\|_{\infty}$, $\|\cdot\|_{\mathrm{F}}$, $\|\cdot\|_*$ and $\|\cdot\|_{\mathrm{op}}$ be the entrywise-maximum, Frobenius, nuclear and operator norms of a matrix, respectively.  Let $\|\cdot\|$ be the $\ell_2$-norm of a vector.  For any matrix $A \in \mathbb{R}^{n \times n}$ and a 0/1-matrix $B \in \{0, 1\}^{n \times n}$, let~$A_B$ and $A_{\overline{B}}$ be the sub-matrices of $A$ indexed by the one and zero entries of $B$, respectively.  For any matrix $M \in \mathbb{R}^{n \times m}$, let $M_i$ be $M$'s $i$th row, $i = 1, \ldots, n$. For two deterministic or random $\mathbb R$-valued sequences $a_n, b_n > 0$, write $a_n \gg b_n$ if $a_n/b_n \to \infty$ as $n$ diverges. Write $a_n \lesssim b_n$ if $a_n/b_n \leq C$, and write $a_n \asymp b_n$ if $c \leq a_n/b_n \leq C$, for some absolute constants $c, C > 0$ for all $n \geq 1$. Denote $\mathbb N$ and $\mathbb N^*$ the set of non-negative integers and the set of natural numbers, respectively.  Let $\odot$ be the matrix Hadamard product operator.   For any two matrices $A$ and $B$ of the same dimensions, where $B$ is a $0/1$ matrix, let $A_B = A \odot B$ denote $A_{ij}$ is observed if and only if $B_{ij} = 1$. Let $\overset{\mathrm{ind}.}{\sim}$ and $\overset{\mathrm{i.i.d.}}{\sim}$ denote the ``independently distributed'' and ``independent and identically distributed'', respectively.

\section{Problem setup}\label{sec-problem-setup}

This paper contains four key ingredients: dynamic networks, missing data analysis, temporal dependence and online change point detection. We unfold the framework concerned in this paper. 

We start by considering the marginal network model. Let $n \in \mathbb N^*$ be the number of nodes. At each time, we assume that the adjacency matrix $Y \in \{0,1\}^{n \times n}$ is marginally generated from a random dot product graph (RDPG) model. The latent positions $\{X_i\}_{i = 1}^n \subset \mathbb R^d$ are assumed to be independent and identically distributed (i.i.d.)~and follow the inner product distribution defined below. Note that the inner product distribution and RDPG refer to the marginal distributions of the latent positions and the adjacency matrix, respectively.

\begin{definition}[Inner product distribution] \label{def-inner-product-dist}
    Let $F$ be a probability distribution whose support is given by $\mathcal{X}_F \subset \mathbb{R}^d$ for $d \in \mathbb N^*$.  We say that $F$ is a $d$-dimensional inner product distribution if for all $x, y \in \mathcal{X}_F$, it holds that $x^{\top}y \in [0, 1]$.
\end{definition}

\begin{definition}[RDPG]
\label{def-rdpg}
Let $\{F_m\}_{m = 1}^M$
be a sequence of $d$-dimensional inner product distributions. Let $F = \sum_{m = 1}^M\tau_m F_m$ with weights $\{\tau_m\}_{m = 1}^M$ such that $\tau_m \geq 0$ and $\sum_{m = 1}^M\tau_m = 1$. The support of $F$ is denoted as $\mathcal{X}_F = \cup_{m = 1}^M\mathcal{X}_{F_m}$. Let $Y \in \{0, 1\}^{n \times n}$ be the adjacency matrix of an RDPG with latent positions $\{X_i\}_{i = 1}^n \overset{\mathrm{i.i.d.}}{\sim} F$, if $Y_{ij}|\{X_i, X_j\} \overset{\mathrm{ind.}}{\sim} \mathrm{Bernoulli}(P_{ij})$, with $P_{ij} = X_i^{\top}X_j$, $i \leq j$.  Denote further the associated graphon matrix as $\Theta = \{\mathbb E(P_{ij})\}_{(i,j) \in \{1,\dots,n\}^{\otimes 2}} \in [0, 1]^{n \times n}$ with rank denoted as $\mathrm{rank}(\Theta)$.
\end{definition} 

Without further assumption on the rank of $\Theta$, it follows from \Cref{def-rdpg} that $\mathrm{rank}(\Theta) \leq d$. This modelling approach based on Definitions \ref{def-inner-product-dist} and \ref{def-rdpg} is general, including stochastic block models as special cases.  We refer readers to \cite{athreya2017statistical} for a comprehensive review.  For brevity, we focus on undirected networks in this paper. As a consequence, the adjacency matrix $Y$, the graphon matrix $\Theta$ and the missingness matrix $\Omega$, the missingness probability $\Pi$, which will be introduced below, are symmetric matrices.  

In practice, large networks often contain missingness.  Formulating networks into adjacency matrices naturally leads to the territory of matrix completion analysis.  In the existing literature, different assumptions have been proposed based on missingness patterns.  Besides structural missingness patterns (for instance, \citealp{cai2016structured} assumed observing a few full rows and columns), two main types of sampling assumptions have been used in the literature: the Bernoulli sampling model (BSM, \Cref{def-sample}) and the uniform sampling at random (USR).  Given the matrix basis $\mathcal{B} = \{E_{kl} = e_k e_l^{\top}\}_{(k, l) \in \{1, \ldots, n\}^{\otimes 2}} \in \{0,1\}^{n \times n}$, where $e_k$ with $(e_k)_j = \mathbbm{1}\{k =j\}$ is the standard basis, the USR assumes that observations are i.i.d.~from~$\mathcal{B}$.  Since under the USR, the probability of observing one entry multiple times is non-zero, we resort to the BSM detailed in \Cref{def-sample} to suit the purpose of this paper.  \Cref{def-sample}, with the latent positions $\{X_i\}_{i = 1}^n$ as the inputs, is reminiscent to the graph associated sampling scheme considered in \cite{bhojanapalli2014universal} and \cite{bhaskar2016probabilistic}.

\begin{definition}[Heterogeneous Bernoulli sampling model] \label{def-sample}
Let $\Omega \in \{0, 1\}^{n \times n}$ be the Bernoulli sampling matrix of an RDPG if $\Omega_{ij} \overset{\mathrm{ind.}}{\sim} \mathrm{Bernoulli}(\Pi_{ij})$, where $\Pi_{ij}$ is the $(i,j)$th entry of the missingness probability matrix as $\Pi = \{\mathbb E(\Omega_{ij})\}_{(i,j) \in \{1, \ldots, n\}^{\otimes 2}} \in [0, 1]^{n \times n}$.
\end{definition}

\begin{remark}\label{remark:var}
    Both Definitions \ref{def-rdpg} and \ref{def-sample} are used to generate symmetric random matrix $V \in \mathbb{R}^{n \times n}$ of $0/1$ entries with a specified mean matrix $\mathbb{E}(V)$, but the statistical properties of the generated matrices are different. For \Cref{def-rdpg}, the entries are only conditionally independent given the latent positions. Unconditionally, two entries sharing the same row or column are dependent, thus the variance statistic $\|\mathbb{E}[\{V-\mathbb{E}(V)\}\{V-\mathbb{E}[
    (V)\}^{\top}]\|_{\op} \lesssim n^2$. Whereas for \Cref{def-sample}, all entries are unconditionally independent thus the variance statistic $\|\mathbb{E}[\{V-\mathbb{E}(V)\}\{V-\mathbb{E}[
    (V)\}^{\top}]\|_{\op} \lesssim n$. Detailed computation can be found in the proof of \Cref{lem-lem4-klopp_dep} in the Appendix.
\end{remark}

Moving on from marginal modelling of a static network to dynamic networks, we let $\{Y(t)\}_{t \in \mathbb N^*}\subset \mathbb \{0,1\}^{n \times n}$ and $\{\Omega(t)\}_{t \in \mathbb N^*}\subset \mathbb \{0,1\}^{n \times n}$ be sequences of adjacency and missingness matrices, respectively, where the number of nodes $n$ is assumed to be fixed across $t$. We allow temporal dependence on both $\{Y(t)\}_{t \in \mathbb N^*}$ and $\{\Omega(t)\}_{t \in \mathbb N^*}$, which are measured by its $\phi$-mixing coefficients defined below.

\begin{definition}[$\phi$-mixing coefficients]\label{def-alpha-mixing}
For a sequence of random objects $\{Z(t)\}_{t \in \mathbb{N}^*}$ and two positive integers $j$ and $\ell$, let $\mathcal{A}_{1}^j$ and $\mathcal{A}_{j+\ell}^{\infty}$ be the $\sigma$-algebras generated respectively by $\{Z(t)\}_{t = 1}^j$ and $\{Z(t)\}_{t = j+\ell}^{\infty}$.  The $\phi$-mixing coefficients of $\{Z(t)\}_{t \in \mathbb{N}^*}$ for any $\ell \in \mathbb{N}^*$ are defined as
    \[
        \phi_Z(\ell) = \sup_{j \in \mathbb{N}^*}\sup_{A \in \mathcal{A}_1^j,\, B \in \mathcal{A}_{j+\ell}^{\infty}}\big|\mathbb{P}(B|A) - \mathbb{P}(B)\big|.
    \]
\end{definition}
The $\phi$-mixing coefficients measure temporal dependence, with measurements including $\alpha$-, $\beta$-mixing coefficients \citep{bradley2005basic} and the functional dependence measure \citep{wu2005nonlinear}.  We adopt the $\phi$-mixing condition in this paper owing to its relevance to the ergodicity of Markov chains, which is a more natural choice for the dependent dynamic models that we will elaborate later. For the sequence of adjancy matrices $\{Y(t)\}_{t \in \mathbb N^*}$, its temporal dependence is induced by that of the sequence of latent positions $\{X_i(t)\}_{i = 1,\, t \in \mathbb N^*}^n$. We denote respectively the $\phi$-mixing coefficients of $\{X_i(t)\}_{t \in \mathbb N^*}$ and $\{\Omega_{ij}(t)\}_{t \in \mathbb N^*}$ as $\phi_{X_i}(\ell)$ and $\phi_{\Omega_{ij}}(\ell)$, $i, j \in \{1, \ldots, n\}$.

With the aforementioned dynamic networks, sampling models and $\phi$-mixing coefficients, we are now ready to unveil the full framework.
\begin{assumption}[Model] \label{assume-model-new}
    Let the process $\{Y_{\Omega}(t), \Omega(t)\}_{t \in \mathbb{N}^*} \subset \mathbb{R}^{n \times n}$ satisfy the following conditions.
    \begin{enumerate}[leftmargin=*,align=left, nolistsep, topsep=0pt]
        \item[a.] For each $t \in \mathbb{N}^*$, $n$ latent positions $\{X_i(t)\}_{i = 1}^n \overset{\mathrm{ind.}}{\sim} F$, where $F$ is a $d$-dimensional inner product distribution defined in \Cref{def-inner-product-dist}. 
        \item[b.] For each $t \in \mathbb{N}^*$, let $Y(t)$ be the adjacency matrix of an RDPG with latent positions $\{X_i(t)\}_{i = 1}^n$, %defined in \Cref{def-rdpg}, 
        denoting its graphon matrix by $\Theta(t) = \mathbb E\{P(t)\}$ with $P_{ij}(t) = \{X_i(t)\}^{\top}X_i(t)$, defined in \Cref{def-rdpg}.  
        Assume that the graphon matrices satisfy $\sup_{t \in \mathbb N^*}\mathrm{rank}\{\Theta(t)\} \leq r$, and 
        \[
            \max_{i = 1}^n \phi_{X_{i}}(\ell) \leq C_{\phi}\rho^{\ell}, \quad \ell \in \mathbb{N}^*,
        \]
        for some constant $\rho \in (0,1)$, and 
        \[
            \sup_{t \in \mathbb{N}^*}\max_{1 \leq i\leq j \leq n} P_{ij}(t) \leq \vartheta < 1.
        \]
        \item[c.] For each $t \in \mathbb{N}^*$, let $\Omega(t)$ be the Bernoulli sampling matrix of an RDPG with the missingness probability matrix $\Pi(t) = \mathbb{E}\{\Omega(t)\}$, defined in \Cref{def-sample}.  
        Assume that 
        $\{\Pi(t)\}_{t \in \mathbb{N}^*}$ satisfies $$0 < q_1 \leq \inf_{t \in \mathbb{N}^*}\min_{1 \leq i \leq j \leq n} \Pi_{ij}(t) \leq \sup_{t \in \mathbb{N}^*}\max_{1 \leq i \leq j \leq n} \Pi_{ij}(t) \leq q_2 \leq 1,$$
        and $\{\Omega(t)\}_{t \in \mathbb{N}^*}$ satisfies
        \[
            \max_{1 \leq i \leq j \leq n}\phi_{\Omega_{ij}}(\ell) \leq C_{\phi}\rho^{\ell}, \quad \ell \in \mathbb{N}^*,
        \]
        for some constant $\rho \in (0,1)$.
        \item[d.] The latent positions $\{X_i(t)\}_{i = 1, t \in \mathbb{N}^*}^n$ and the missingness matrices $\{\Omega(t)\}_{t \in \mathbb{N}^*}$ are independent.
    \end{enumerate}
\end{assumption}

\Cref{assume-model-new}b.~assumes the exponential decay of the $\phi$-mixing coefficients of the latent positions.  The temporal dependence of $\{X_i(t)\}_{i = 1, \, t \in \mathbb{N}^*}^n$ induces the temporal dependence of $\{Y(t)\}_{t \in \mathbb N^*}$. This statement is formalised in the next lemma.

\begin{lemma}\label{theorem:markov-phi}
    Suppose that \Cref{assume-model-new} holds.  For any $(i,j) \in \{1, \ldots, n\}^{\otimes 2}$, let $\{\phi_{Y_{ij}}(\ell)\}_{\ell \in \mathbb{N}}$ denotes the $\phi$-mixing coefficient sequences of $\{Y_{ij}(t)\}_{t \in \mathbb{N}^*}$. It then holds that $\phi_{Y_{ij}}(\ell) \leq 2C \rho^{\ell}$, $\ell \in \mathbb{N}^*$, with the $\rho$ and $C$ in \Cref{assume-model-new}b.
\end{lemma}
The proof of \Cref{theorem:markov-phi}, along with all the other technical details of this paper, is collected in Appendix.  To accommodate the change points, we introduce the following two assumptions focusing on two scenarios.

\begin{assumption}[No change point] \label{assume-no-change}
Assume that $\Theta(1) = \Theta(2) = \cdots$.
\end{assumption}

\begin{assumption}[One change point] \label{assume-one-change}
Assume that there exists a positive integer $\Delta \geq 1$ such that $\Theta(1) = \cdots \Theta(\Delta) \neq \Theta(\Delta + 1) = \cdots$.
\end{assumption}

In view of Assumptions~\ref{assume-model-new}, \ref{assume-no-change} and \ref{assume-one-change}, we are concerned with detecting change points in the underlying network distributions, but not the missingness mechanisms.  In fact, we allow the missingness probability matrix $\Pi(t)$ defined in \Cref{def-sample} to be distinct at every time point.  The heterogeneity of the missingness probability matrices, undoubtedly, increases the difficulty in detecting the changes in the graphon changes.

From the online change point detection point of view, our task is to seek $\widehat{\Delta}$, an estimator of~$\Delta$, such that the following hold.  The quantities $\mathbb{P}_{\infty}(\cdot)$ and $\mathbb{E}_{\infty}(\cdot)$ indicate the probability and expectation under \Cref{assume-no-change}, with $\mathbb{P}_{\Delta}(\cdot)$ and $\mathbb{E}_{\Delta}(\cdot)$ denoting their counterparts under \Cref{assume-one-change}.

\begin{enumerate}[leftmargin=*,align=left, nolistsep, topsep=0pt]
    \item [$\bullet$] The overall false alarm control.  With a pre-specified $\alpha > 0$, it holds $\mathbb{P}_{\infty}\{\exists t \in \mathbb{N}^*: \, \widehat{\Delta} \leq t\} < \alpha$.  
    \item [$\bullet$] The detection delay control.  If $\Delta < \infty$, then $\mathbb{P}_{\Delta}\{\Delta \leq \widehat{\Delta} \leq \Delta + \varpi\} > 1 - \alpha$, where $\varpi$ is referred to as detection delay and is to be minimised. 
\end{enumerate}

We remark that in the online change point detection literature, there are two types of control.  In addition to the false alarm controls, it is also popular to lower bound $\mathbb{E}_{\infty}[\widehat{\Delta}]$ - the expectation of the change point location estimator under \Cref{assume-no-change}, see e.g.~\cite{lorden1971procedures}, \cite{lai1995sequential} and \cite{lai1998information}.  These two types of control share many similarities but with some subtle differences: controlling the overall false alarm probability is more preferable in theoretical analysis and controlling the average run length is handier in practice.  We refer readers to \cite{yu2020note} for detailed discussions on this matter and stick with controlling the overall false alarm probability in this paper.

\section{Optimal network change point detection with missingness and dependence}\label{sec-main}

Given the model described in \Cref{assume-model-new}, for the task of online change point detection, in this section, we propose a soft-impute-based change point estimator, with nearly-optimal theoretical guarantees.  The sub-routine soft-imputation is presented in \Cref{sec-soft-impute}, the change point analysis is conducted in \Cref{sec-change-point}, the fundamental limits are discussed in \Cref{sec-minimax} and a special example on a Markov chain-type dynamic network is investigated in \Cref{sec:example}.

\subsection{The soft-impute algorithm}\label{sec-soft-impute}

As for the task of matrix completion, a rich collection of estimators have been proposed, most of which are penalisation-based methods working under some form of sparsity assumptions.  The soft-impute algorithm \citep[e.g.][]{mazumder2010spectral} iteratively solves a nuclear norm penalised least squares estimator.  We present a multiple-copy version of the soft-impute algorithm, which is a variant of the soft-impute algorithm studied in \cite{klopp2015matrix}.

Given any integer pair $0 \leq s < e$, let 
    \begin{align}\label{eq-obj-function-hatM}
        \widetilde{\Theta}_{s:e} \in \argmin_{\Theta \in \mathbb{R}^{n \times n}}f\left(\Theta; \{Y_{\Omega}(i), \Omega(i)\}_{i = s+1}^e, \widehat{\Theta}, \lambda_{s, e}\right),
    \end{align}
    where $f(\Theta; \{Y_{\Omega}(i), \Omega(i)\}_{i = s+1}^e, \widehat{\Theta}, \lambda_{s, e}) = \{2(e-s)\}^{-1} \sum_{i = s+1}^e \big\|Y_{\Omega}(i) + \widehat{\Theta}_{\overline{\Omega(i)}} - \Theta\big\|^2_{\mathrm{F}} + \lambda_{s, e} \|\Theta\|_*$, $\lambda_{s, e} > 0$ is a tuning parameter and $\widehat{\Theta} \in \mathbb{R}^{n \times n}$ is an optimiser obtained from the previous step or the initialiser of the algorithm.  In fact, $\widehat{\Theta}$ can be any dimensional-compatible matrix. Then, $\widehat{\Theta}_{s:e}$ is obtained by truncated each entry of $\widetilde{\Theta}_{s:e}$ by another tuning paramter $a \in (0,1)$.

%\pd{In this algorithm, $\lambda$, $a$, $\tilde{\Theta}$ and $\hat{\Theta}$ must be indexed by $s:e$? In the previous paragraph should we change $s:t$ to $s:e$ to be consistent in notation, as $t$ is also used to index the "time" in the dynamic networks.}

\begin{algorithm}
	\begin{algorithmic}
		\INPUT $\{Y_{\Omega}(t), \Omega(t)\}_{t = s+1, \ldots, e}$, $\lambda_{s, e}, a > 0$, $\widehat{\Theta} \in \{0,1\}^{n \times n}$
		\State $\mathrm{FLAG} \leftarrow 0$
			\While{$\mathrm{FLAG} = 0$}
				\State{
				\begin{equation}\label{eq-soft-impute}
				    \widetilde{\Theta} \leftarrow S_{\lambda_{s, e}} \bigg\{(e-s)^{-1}\sum_{t = s+1}^e  \left[Y_{\Omega}(t)  + \widehat{\Theta}_{\overline{\Omega(t)}})\right]\bigg\}
				\end{equation}}
				\If{$\left\|(e-s)^{-1} \sum_{t = s+1}^e (\widetilde{\Theta} - \widehat{\Theta})_{\overline{\Omega(t)}}\right\|_{\mathrm{op}} < \lambda_{s, e}/3$ and $\|\widetilde{\Theta} - \widehat{\Theta}\|_{\infty} < a$}
					\State $\mathrm{FLAG} \leftarrow 1$
				\EndIf
				\State $\widehat{\Theta}_{ij} \leftarrow \mathrm{sign}(\widetilde{\Theta}_{ij}) \min\{|\widetilde{\Theta}_{ij}|, \, a\}$
			\EndWhile
		\OUTPUT $\widehat{\Theta}_{s:e} = \widehat{\Theta}$
		\caption{Soft-Impute\label{alg-main-klopp}}
	\end{algorithmic}
\end{algorithm}
%\pd{how about moving this definition before the algorithm?}
The soft-thresholding estimator $S_{\lambda}(\cdot)$ in \eqref{eq-soft-impute} is defined as follows.  For any matrix $W \in \mathbb{R}^{n \times n}$, let $S_{\lambda}(W) = UD_{\lambda} V^{\top}$, where $W = UDV^{\top}$ is a singular value decomposition of $W$, $D = \mathrm{diag}\{d_i, i = 1, \ldots, n\}$ is a diagonal matrix containing all singular values of $W$ including zero and $D_{\lambda} = \mathrm{diag}\{(d_i - \lambda)_+, i = 1, \ldots, n\}$, where $(\cdot)_+ = \cdot \vee 0$.

For any integer pair $0 \leq s < e$, we solve \eqref{eq-obj-function-hatM} using \Cref{alg-main-klopp}, with pre-specified $\lambda_{s,e}, a >0$.  \Cref{alg-main-klopp} terminates when the convergence criteria are met.  When $e-s = 1$, it is well-understood that \eqref{eq-soft-impute} is the minimiser of \eqref{eq-obj-function-hatM} and the iterative algorithm converges \citep{mazumder2010spectral, klopp2015matrix}.  As for the multiple-copy version that we exploit in this paper, we provide a sanity check in \Cref{lem-soft-impute-property}.  We would like to emphasise that \Cref{lem-soft-impute-property} is a deterministic result, working with any given input matrices.

\begin{lemma}\label{lem-soft-impute-property}
For any $m \in \mathbb{N}^*$, any sequence of matrices $\{R(t)\}_{t = 1}^m$ and $\lambda > 0$, the solution to the optimisation problem $\min_Z \left\{(2m)^{-1} \sum_{t = 1}^m \|R(t) - Z\|^2_{\mathrm{F}} + \lambda \|Z\|_*\right\}$
%\[
%\min_Z \left\{(2m)^{-1} \sum_{t = 1}^m \|R(t) - Z\|^2_{\mathrm{F}} + \lambda \|Z\|_*\right\}
%\]
is given by $\widehat{Z} = S_{\lambda} \{m^{-1} \sum_{t = 1}^m R(t) \}$.
    
Let $\{M_k\}_{k \in \mathbb{N}}$ be the sequence of solutions produced by \Cref{alg-main-klopp}, with input indexed by $t \in \{1, \ldots, m\}$.  We have that $\max\{\|M_{k+1} - M_k\|_{\mathrm{F}}$, $\|M_{k+1} - M_k\|_{\infty}\} \to 0$ as $k \to \infty$.  For any $t \in \{1, \ldots, m\}$, $\|(M_{k+1} - M_k)_{\overline{\Omega(t)}}\|_{\mathrm{op}} \to 0$ as $k \to \infty$.
\end{lemma}

\subsection{The change point detection algorithm}\label{sec-change-point}

With the soft-impute estimator in hand, we define the change point estimator as
    \begin{equation}\label{eq-widehat-delt-defi}
        \widehat{\Delta} = \min\left\{t \geq 2: \, \max_{s = 1}^{t-1} (\widehat{D}_{s, t} - \varepsilon_{s, t}) \geq 0\right\},
    \end{equation}
    where $\widehat{D}_{s, t} = \|\widehat{\Theta}_{0:s} - \widehat{\Theta}_{s:t}\|_{\mathrm{F}}$ for any pair of positive integers $s, t$, $t \geq 2$ and $\{\varepsilon_{s, t}\}$ is a sequence of pre-specified tuning parameters.  The search of the change point is terminated once a change point is detected, or there is no more new data point observed.  

    \begin{remark}[Computational complexity]
        To obtain our change point estimator $\widehat{\Delta}$, we need a sequence of matrix estimators $\{\widehat{\Theta}_{s:t}\}$, output by \Cref{alg-main-klopp}.  The computational cost includes that of conducting a truncated singular value decomposition and reconstructing a low-rank matrix, and depends on the iterations, warm-start choices, etc.  We omit the details here and refer readers to Section~5 in \cite{mazumder2010spectral}.  On top of the matrix estimation, we also have the change point analysis procedure.  In \eqref{eq-widehat-delt-defi}, for notational simplicity, we present a rather expensive procedure exploiting all integer pairs $s < t$.  In fact, without any loss of statistical accuracy, one can instead exploit a dyadic grid $s \in \{(t - 2^j) \vee 1\}_{j \in \mathbb{N}^*}$, for any $t \geq 2$, see \cite{yu2020note} for detailed discussions.  This means, for every newly collected time point $t \geq 2$, the additional cost of estimating the change point is of order $O\{\log(t) \mathcal{Q}\}$, where $\mathcal{Q}$ is the cost of \Cref{alg-main-klopp}.  We comment that the logarithmic cost of online change point detection cannot be improved, without knowing the exact pre-change distribution.
    \end{remark}

    To understand the theoretical performances of our change point estimator $\widehat{\Delta}$, we require the following signal-to-noise ratio condition.

\begin{assumption}[Signal-to-noise ratio] \label{assume-snr}
Under \Cref{assume-one-change}, let $\kappa = \|\Theta(\Delta) - \Theta(\Delta+1)\|_{\mathrm{F}}$.  For a pre-specified $\alpha \in (0, 1)$, it satisfies that 
\begin{align*}
    \kappa^2 \Delta \geq C_{\mathrm{SNR}}\frac{r\vartheta n\log(n)\log(\Delta/\alpha)}{q_1^{2}}\max\{1, nq_2^2\},
\end{align*}
where $C_{\mathrm{SNR}} > 0$ is a large enough absolute constant.
\end{assumption}

We refer to \Cref{assume-snr} as the signal-to-noise ratio condition, in the sense that the signal strength is completely characterised by the jump size $\kappa$ and the pre-change sample size $\Delta$, and the noise level is quantified by the network size $n$, network entrywise-sparsity parameter $\vartheta $, network low-rank sparsity $r$ and the strength of missingness $q_1, q_2$.  When $q_1 = q_2 =1$, i.e.~without missingness, up to logarithmic factors, \Cref{assume-snr} has the same form as what is required for a polynomial-time algorithm to yield a nearly-optimal detection delay, in an online network change point detection problem without missing entries \citep{yu2021optimal}.

\begin{theorem}\label{thm-upper-bound}
Consider the model described in \Cref{assume-model-new}.  Let $\alpha \in (0, 1)$ and $\widehat{\Delta}$ be defined in~\eqref{eq-widehat-delt-defi}, with the inputs of \Cref{alg-main-klopp} being $a = \vartheta$, for any $s, t \in \mathbb N$, $s < t$, the penalisation levels being
\begin{align*} 
\lambda_{s,e} = C_{\lambda}\sqrt{\frac{\log(n) + \log(1/\alpha)}{e-s}}\sqrt{\vartheta}\left(\sqrt{n} + q_2n\right),
 \end{align*}
and with the change point thresholds being 
    \begin{align*}
        \varepsilon_{s, t} = 
        C_{\varepsilon} \frac{\sqrt{ r\vartheta n\log(n)\max\{1, nq_2^2\}} }{q_1}\left(\sqrt{\frac{\log(s/\alpha)}{s}} + \sqrt{\frac{\log(t/\alpha)}{t-s}}\right)
    \end{align*}
    where $C_{\lambda}, C_{\varepsilon} > 0$ are sufficiently large absolute constants.\\
If \Cref{assume-no-change} holds, then we have that $\mathbb{P}\{\exists t \in \mathbb{N}^*: \, \widehat{\Delta} \leq t\} \leq \alpha$. 
If Assumptions~\ref{assume-one-change} and \ref{assume-snr} hold, then we have that $\mathbb{P}\left\{\Delta \leq \widehat{\Delta} \leq \Delta + \varpi \right\} \geq 1 - \alpha$, where  
\begin{align*}
    \varpi = C_{\varpi}  \frac{r\vartheta n\log(n)\log(\Delta/\alpha)}{q_1^{2}\kappa^2}\max\{1, nq_2^2\}
\end{align*}
and $C_d > 0$ is a sufficiently large absolute constant.
\end{theorem}

\begin{remark}[The hierarchy of constants.]
    When there exists a change point, i.e.~under Assumption~\ref{assume-one-change}, we can see from \Cref{thm-upper-bound} that there are four absolute constants $C_{\mathrm{SNR}}$, $C_{\lambda}$, $C_{\varepsilon}$ and $C_{\varpi}$.  We do not claim optimality for these constants, but only show here that the feasible choices of these constants do not form an empty set.  One first needs to establish an absolute constant $C_{\mathrm{noise}} > 0$ appearing in \Cref{cor-main-klopp-cor-3}.  The choice of $C_{\lambda}$ solely depends on $C_{\mathrm{noise}}$, $C_{\varepsilon}$ only depends on $C_{\lambda}$, and then $C_{\varpi}$ only depends on $C_{\varepsilon}$.  Once $C_{\varepsilon}$ and $C_{\varpi}$ are chosen, one only needs to ensure that $C_{\mathrm{SNR}}$ is large enough -- larger than some increasing functions of $C_{\varepsilon}$ and $C_{\varpi}$. 
\end{remark}

It is shown in \Cref{thm-upper-bound} that, with the overall false alarm probability upper bounded by pre-specified $\alpha$, if a change point does occur, our procedure will, with probability at least $1 - \alpha$, detect the change point with a detection delay upper bounded by 
    \begin{equation}\label{eq-detection-delay-rate}
        C_\varpi  \frac{r\vartheta n\log(n)\log(\Delta/\alpha)}{q_1^{2}\kappa^2}\max\{1, nq_2^2\}.
    \end{equation}
The derivation of \Cref{thm-upper-bound} relies on two key ingredients: an estimation error bound on imputing a low-rank matrix and an online change point analysis procedure.  To further understand \Cref{thm-upper-bound}, we compare our result with \cite{klopp2015matrix} and \cite{yu2021optimal}.  Before delving into details, we highlight the dependence we allow.  Compared to \cite{klopp2015matrix}, we allow for the dependence within network edges and missingess, as well as the dependence between the presences and missingness of edges.  This is materalised by the interconnection via random latent positions.  Compared to \cite{yu2021optimal}, we allow for the temporal dependence across time and this is formalised by the $\phi$-mixing coefficients.

\textbf{Estimation error analysis.} Our estimation error analysis is built upon the analysis in \cite{klopp2015matrix}, which presents an estimation error bound of a soft-impute estimator based on only one incomplete matrix.  To deal with the dynamic networks scenario, we derive a multiple-copy version of Corollary~3 in \cite{klopp2015matrix} in \Cref{cor-main-klopp-cor-3}, under temporal dependence and network dependence. We provide detailed analysis in \Cref{remark:error_analy}

\begin{proposition} \label{cor-main-klopp-cor-3}
Assume that Assumptions~\ref{assume-model-new} and \ref{assume-no-change} hold.  Denote $\Theta(t) = \Theta_0$, $t \in \mathbb{N}^*$, with $\mathrm{rank}(\Theta_0) \leq r$.  For any integer pair $0 \leq s < e$, let $\widehat{\Theta}_{s:e}$ be the output of \Cref{alg-main-klopp} with $a = \vartheta$ and 
\begin{align*}
\lambda_{s,e} = C_{\lambda}\sqrt{\frac{\log(n) + \log(1/\delta)}{e-s}}\sqrt{\vartheta}\left(\sqrt{n} + q_2n\right).
\end{align*}
For any interval $(s,e]$, such that the interval length $(e-s)$ satisfying \begin{align}\label{eq:interval_leng}
    \frac{\sqrt{e-s}}{\log^2(e-s)} \geq C\sqrt{n\log(n)}.
\end{align}
Then with probability at least $1 - 2\delta$, it holds that 	\begin{align}\label{eq-estimation-error-bounds}
	\|\widehat{\Theta}_{s:e} - \Theta_0\|_{\mathrm{F}}^2 \leq \frac{C_{\mathrm{noise}}\vartheta rn\log(n)\log(1/\delta)}{q_1^2(e-s)}\max\left\{1, q_2^2n\right\},
	\end{align}
	where $C_{\lambda}, C, C_{\mathrm{noise}} > 0$ are absolute constants.
\end{proposition}

\begin{remark}\label{remark:error_analy}
 \ \\
 (i) \textbf{Cost of within-network dependence.} The factor $n \max\{1, q_2^2n\} = \max\{n, q_2^2 n^2\}$ in \eqref{eq-estimation-error-bounds} contains two terms reflecting the variation of $\Omega(t)$ and $\Omega(t)\odot\{Y(t) - \Theta(t)\}$, respectively, as discussed in \Cref{remark:var}. Given the within network dependence of $Y(t)$ induced by the latent positions, the term $q_2^2n^2$ is unavoidable and can be understood as the cost of the model generality.  \cite{yu2021optimal} consider the inhomogeneous Bernoulli networks forcing the upper diagonal entries of $Y(t)$ to be unconditionally independent, which reduces $n^2$ to $n$.
 
\noindent (ii) \textbf{Cost of temporal dependence.} We also allow for the temporal dependence in both $\{Y(t)\}_{t \in \mathbb{N}^*}$ and $\{\Omega(t)\}_{t \in \mathbb{N}^*}$. The cost of temporal dependence is reflected in the term $\log^2(e-s)$ in \eqref{eq:interval_leng}, which follows from an application of a variant of matrix Bernstein inequality under temporal dependence (see \Cref{lemma:matrix_bernstein}). 

\noindent (iii) \textbf{Cost of missingness.} This is reflected by the inflation factor $q_1^{-2}$. When there is no missingness, we have $q_1 = q_2 = 1$. Compared to \cite{klopp2015matrix}, since we deal with multiple copies of input matrices, the term $(e-s)^{-1}$  incurrs.
\end{remark}

\textbf{Online change point analysis.}  \cite{yu2021optimal} studies an optimal online network change point detection without missingness and under the inhomogeneous Bernoulli networks.  To compare with \cite{yu2021optimal}, we note that when $q_1 = q_2 = 1$ and the upper diagonal entries of $Y(t)$ are conditionally independent, \eqref{eq-detection-delay-rate} corresponds to the detection delay rate without missingness that
    \begin{equation}\label{eq-detection-delay-without-missingness}
        rn\vartheta\log(\Delta/\alpha) \kappa^{-2}.
    \end{equation}
    It is shown in \cite{yu2021optimal} that, without missingness, a minimax lower bound on the detection delay is of order $\max\{r^2/n, \, 1\} n \vartheta \log(1/\alpha)  \kappa^{-2}$, which can be obtained, off by a logarithmic factor, by an NP-hard algorithm.  As for polynomial-time algorithms, \cite{yu2021optimal} shows that without missingness, the detection delay is exactly that in \eqref{eq-detection-delay-without-missingness}.

%\pd{In the interest of space, should we move Computational complexity to the appendix?}
%\textbf{Computational complexity.} To obtain our change point estimator $\widehat{\Delta}$, we need a sequence of matrix estimators $\{\widehat{\Theta}_{s:t}\}$, output by \Cref{alg-main-klopp}.  The computational cost includes that of conducting a truncated singular value decomposition and reconstructing a low-rank matrix, and depends on the iterations, warm-start choices, etc.  We omit the details here and refer readers to Section~5 in \cite{mazumder2010spectral}.  On top of the matrix estimation, we also have the change point analysis procedure.  In \eqref{eq-widehat-delt-defi}, for notational simplicity, we present a rather expensive procedure exploiting all integer pairs $s < t$.  In fact, without any loss of statistical accuracy, one can instead exploit a dyadic grid $s \in \{(t - 2^j) \vee 1\}_{j \in \mathbb{N}^*}$, for any $t \geq 2$, see \cite{yu2020note} for detailed discussions.  This means, for every newly collected time point $t \geq 2$, the additional cost of estimating the change point is of order $O\{\log(t) \mathcal{Q}\}$, where $\mathcal{Q}$ is the cost of \Cref{alg-main-klopp}.  We comment that the logarithmic cost of online change point detection cannot be improved, without knowing the exact pre-change distribution.

\textbf{Choices of tuning parameters.} The tuning parameters involved in the procedure are $\lambda_{s, e}$ - the penalisation level and the convergence indicator when calculating $\widehat{\Theta}_{s, e}$ using \Cref{alg-main-klopp}, $a$ - the truncation parameter used in \Cref{alg-main-klopp} tailored for the sparsity of networks and $\{\varepsilon_{s, t}\}$ - the sequence of change point thresholds.  To guarantee the theoretical results, all these tuning parameters involve unknown model parameters.  The practical guidance is left to \Cref{sec-numerical}.   We are to explain their theoretical choices here.  

As for $\lambda_{s, e}$, it is chosen to be $\lambda_{s,e} \asymp \|\frac{1}{e-s} \sum_{t = s+1}^e \Omega_{ij}(t)\{Y_{ij}(t) - \Theta_{ij}(t)\} E_{ij}\|_{\mathrm{op}}$,
where $E_{ij} = e_ie_j^{\top}$ is a matrix basis.
That is to say, $\lambda_{s, e}$ is chosen to overcome the noise matrix with observed entries.  We do not have guarantees if there is underestimation on $\lambda_{s,e}$, but an over-estimated $\lambda_1 \geq \lambda_{s, e}$ still leads to sufficient guarantees by inflating estimation error bound \eqref{eq-estimation-error-bounds} that 
\[
\|\widehat{\Theta}_{s:e} - \Theta_0 \|^2_{\mathrm{F}} \lesssim \frac{\lambda^2_1 \vee \left\{r\vartheta n \log(n) \log(1/\delta)\max\left\{1, nq_2^2\right\}\right\}}{q_1^{2} (e-s)},
\]
which, with a correspondingly adjusted threshold, directly implies the same false alarm guarantees and an inflated detection delay.  The truncating parameter $a$ is chosen to be the same as the entrywise-sparsity level of the underlying networks.  We do not have theoretical guarantees when $a$ is under-estimated, but when it is over-estimated (in fact, $a \leq 1$ always holds), with a correspondingly adjusted threshold, directly implies the same false alarm guarantees and an inflated detection delay.  The change point thresholds are set to be the minimum order to cast $\max_{s, t} \|\widehat{\Theta}_{0:s} - \widehat{\Theta}_{s:t}\|_{\mathrm{F}}$.  Under- and over-estimation of the change point thresholds leads to a distortion on the overall false alarm probability and detection delay controls, respectively.

\subsection{The fundamental limits}\label{sec-minimax}

It is studied in \cite{yu2021optimal} that without missingness or any form of temporal and within network dependence, a minimax lower bound on the detection delay is of the form
\begin{align}\label{eq-minimiax-lower-bound}
    \max\{r^2/n, \, 1\} n \vartheta \log(1/\alpha)  \kappa^{-2}.
\end{align}
    The lower bound \eqref{eq-minimiax-lower-bound}, up to a logarithmic factor, can be matched by an NP-hard algorithm with the upper bound $\log(\Delta/\alpha) \kappa^{-2}n\vartheta\max\{r^2/n, \, \log(r)\}$.  %A statistical and computational trade-off is observed, which means that the best-known upper bound by a polynomial-time algorithm is of order $\log(\Delta/\alpha) r/(\kappa^2_0 n \rho)$.  

With the presence of missingness, it is equivalent to consider the variance of each entry is inflated from $\vartheta $ to $\vartheta/q_1^2$, where $q_1 \in [0,1]$ denotes the strength of non-missingness and where we assume a homogeneous missingness pattern for simplicity.  This is exactly the same treatment when constructing the minimax lower bound in the matrix completion literature \citep[e.g.][]{koltchinskii2011nuclear}. 

With the presence of within network dependence of $Y(t)$, the variance statistic is inflated from $n$ to $q_2^2n^2$ (see the discussion in Remarks \ref{remark:var} and \ref{remark:error_analy}). 

We therefore obtain a minimax lower bound on the detection delay based on \eqref{eq-minimiax-lower-bound} and replacing $\vartheta $ with $\vartheta  /q_1^2$ and $n$ with $q_2^2n^2$, i.e.
\begin{equation}\label{eq-lb-rate}
    \log(1/\alpha) \max\{r^2/(q_2^2n^2), \, 1\} \kappa^{-2}q_2^2n^2 \vartheta q_1^{-2}.
\end{equation}
Under the homogeneous missingness assumption and RDPG model, the detection delay in \Cref{thm-upper-bound}, which is based on a polynomial-time algorithm, is of order
\begin{equation}\label{eq-ub-rate}
    \log(\Delta/\alpha) r\kappa^{-2} q_2^2n^2\vartheta q_1^{-2}.
\end{equation}
Comparing \eqref{eq-lb-rate} and \eqref{eq-ub-rate}, we see that when $r \asymp 1$, i.e.~in a low-rank regime, up to a logarithmic factor, the detection delay in \Cref{thm-upper-bound} is optimal.  %When $r \gtrsim 1$, we see that there exists a gap between \eqref{eq-lb-rate} and \eqref{eq-ub-rate}, which coincides with the statistical and computational trade-off phenomenon in the online change point detection without missingness.  Such a trade-off is commonly-observed in the high-dimensional literature \citep[e.g.][]{zhang2013communication, loh2015regularized}.

\subsection{An example: A Markov chain-based RDPG model}\label{sec:example}
We conclude this section with some further justification of the $\phi$-mixing conditions we imposed in \Cref{assume-model-new}.  The sticky dynamic network process studied in \cite{padilla2019change} is defined based on RDPG, with the latent position $X_i(t) \in \mathbb{R}^d$ satisfying that
\begin{equation}\label{eq:markov_latentpostion}
    X_i(t) \begin{cases}
        = X_i(t-1), & \text{with probability $\rho$},\\
        \overset{\mathrm{ind.}}{\sim} F, & \text{with probability $1-\rho$},
    \end{cases}
\end{equation}
where $\rho \in [0,1)$ and $F$ is an inner product distribution on $\mathbb{R}^d$ defined in \Cref{def-inner-product-dist}.  The latent positions process given in \eqref{eq:markov_latentpostion} is in fact a Markov chain and satisfied \Cref{assume-model-new}b, as suggested below.
\begin{lemma}\label{lemma:ergodic}
    For any $i \in \{1, \ldots, n\}$, the latent position sequence $\{X_{i}(t)\}_{t \in \mathbb{N}^*}$, given in \eqref{eq:markov_latentpostion} with $\rho \in [0,1)$, is a geometric ergodic Markov chain, and its $\phi$-mixing coefficients satisfy $\phi_{X}(\ell) \leq \rho^{\ell}$, i.e.~$\phi_{X}(\ell)$ decays exponentially in $\ell$. 
\end{lemma}

With \Cref{lemma:ergodic} at hand, we see that the model studied in \cite{padilla2019change} is a special case of the model we considered in this paper.

\section{Numerical experiments} \label{sec-numerical}    

We now illustrate the numerical efficacy of our proposed online change point detection estimator.  Our proposed methods are implemented in the R package \texttt{changepoints} \citep{changepoints_R}.

\subsection{Synthetic data analysis}\label{sec-simulation}

\medskip
\noindent \textbf{Competitors.}  Three competitors are considered.  

\textbf{1}. Instead of using the soft-impute algorithm for matrix completion, we consider the Rank-Constrained Maximum Likelihood Estimator (RC) \citep{cai:2013,bhaskar2016probabilistic}.  To be specific, we replace $\widehat{D}_{s,t}$ in \eqref{eq-widehat-delt-defi} with $\widehat{D}_{s, t}^{\mathrm{MLE}} = \|\widehat{\Theta}_{0:s}^{\mathrm{MLE}} - \widehat{\Theta}_{s:t}^{\mathrm{MLE}}\|_{\mathrm{F}}$, where for $e > s$, 
\[
    \widehat{\Theta}_{s:e}^{\mathrm{MLE}} = \argmin_{\|\Theta\|_{\infty} \leq a, \mathrm{rank}(\Theta) \leq r} F_{s:e}(\Theta), 
\]
with $a > 0$, $r \in \mathbb{N}$ and $F_{s:e}(\Theta)$ being the negative log-likelihood of $\{Y_{\Omega}(t), \Omega(t)\}_{s+1}^e$.  We compute~$\widehat{\Theta}_{s:e}^{\mathrm{MLE}}$ using the approximate projected gradient method \citep[see Algorithm 1 in ][]{bhaskar2016probabilistic}, with the rank constraint $r$ the dimension of the matrix as done in \cite{bhaskar2016probabilistic} and the sup-norm constraint $a = 0.99$. The step size $\tau$ of the approximate projected gradient method is selected via a backtracking line search.% using Armijo's rule.

\textbf{2.} Instead of using the $\widehat{D}_{s,t}$-type scan-statistics to detect change points, we deploy the $k$-nearest neighbour ($k$-NN) based procedure \citep{chu:18,chen:19}, and still call \Cref{alg-main-klopp} for matrix completion.  Three types of statistics are considered in the R package \texttt{gStream} \citep{chen2019package}: the original edge-count scan statistic (ORI), the weighted edge-count scan statistic (W) and the generalised edge-count scan statistic (G).  For fair comparisons, we calibrate the  thresholds thereof on training data so that the false alarm rate is controlled at the level $\alpha$ instead of the average run length as proposed in \cite{chen:19}.  In the simulations we set $k=1$.

\textbf{3.} We consider an online adaptation of the MissInspect change point estimator \citep{foll:21}, which is applied to the vectorised networks with missing elements.  We remark that this might not be a fair comparison due to the different structural assumptions.

\medskip
\noindent \textbf{Tuning parameters.}  To choose $\{\lambda_{s,t}, \varepsilon_{s, t}\}$, each experiment comprises of a calibration step carried out on training data $\{Y_{\Omega}(t), \Omega(t)\}_{t = 1}^{T_{\mathrm{train}}}$, $T_{\mathrm{train}}=200$.  The training data do not possess change points.  

To be specific, for $\{\lambda_{s,t}\}$ we set the sparsity parameter $\vartheta $ of the adjacency matrices as the $95\%$ quantile of $\{\widehat{\Theta}_{0:T_{\mathrm{train}}, ij}\}_{i,j}$, $a = 1$ in \Cref{alg-main-klopp}, the maximum (minimum) missingness probability $q_2$ ($q_1$) as the $95\%$ ($5\%$) quantile of $\{T^{-1}_{\mathrm{train}}\sum_{t = 1}^{T_{\mathrm{train}}}\Omega(t)\}_{i,j}$, the rank $r$ as the rank of $\widehat{M}_{0:T_{\mathrm{train}}}$ and $C_\lambda=2/3$ in all numerical experiments. The robustness of the choice of $C_\lambda$ is investigated in an additional simulation provided in \Cref{sec-addition}.

As for $\{\varepsilon_{s,t}\}$, we randomly permute the training data $K = 100$ times with respective to $t$, and compute the $K$ replicates of the CUSUM statistic in \eqref{eq-widehat-delt-defi}, where for each $s$ and $t$, $\Vert\widehat{D}_{s,t}^{(k)}\Vert_{\mathrm{F}}$ involves only the $k$th permutation, $k \in \{1, \ldots, K\}$.  We choose $C_{\varepsilon}$ such that the proportion of $\{\Vert\widehat{D}_{s,t}^{(k)}\Vert_{\mathrm{F}}\}_{k=1}^K$ which cross the detection threshold is capped at $\alpha$. 

The detection thresholds for the competitors are also calibrated in a similar manner using the training data.  % $\{Y(t), \Omega(t)\}_{t = 1}^{T_{\mathrm{train}}}$.

\medskip
\noindent \textbf{Evaluation measurements.}  Two metrics are considered: the average detection delay 
\[
    \mathrm{Delay} = \left(\sum_{j = 1}^N\mathbbm{1}\{\tilde{\Delta} \geq \Delta\}\right)^{-1} {\sum_{j = 1}^N (\tilde{\Delta} - \Delta)\mathbbm{1}\{\tilde{\Delta} \geq \Delta\}}
\]
%\[%begin{align}\label{eq:delay}
%    \mathrm{Delay} = \big(\sum_{j = 1}^N\mathbbm{1}\{\tilde{\Delta} \geq \Delta\}\big)^{-1} {\sum_{j = 1}^N (\tilde{\Delta} - \Delta)\mathbbm{1}\{\tilde{\Delta} \geq \Delta\}}
%\]%end{align}
and the proportion of false alarms 
\[
    \mathrm{PFA} = N^{-1} \sum_{j = 1}^N\mathbbm{1}\{\tilde{\Delta} < \Delta\},
\]
where $N = 100$ represents the number of repetitions we conduct in each setting, $\tilde{\Delta}$ is the minimum of the total time length and the estimated change points to account for the situations when no change point is detected.

%,  we demonstrate the numerical performance of our method in the setting of a low rank stochastic block model, with different degrees of missingness. We compare the average detection delay and the false alarm rate of Algorithm \ref{alg-main-klopp} against those obtained using three other methods:  In the second scenario, in Section~\ref{sec: scenario2}, we generate dynamic random dot product graphs with fixed latent node positions, and repeat the comparison of our method against the above three methods for a given proportion of missingness. 

%For the false alarm probability $\alpha \in \{0.01, 0.05\}$, 

%The purpose of the calibration step is to obtain data-adaptive estimates of the tuning parameters $\lambda_{s,t}$ and the detection thresholds $\varepsilon_{s, t}$, which will allow us to control the false alarm rate at a desired level $\alpha$ as per Theorem \ref{thm-upper-bound}. We conduct our experiments with $\alpha \in \{0.01, 0.05\}$. 

\medskip
\noindent \textbf{Change points.}  We generate a sequence of network data $\{Y_{\Omega}(t), \Omega(t)\}_{t = 1}^T$ each time, with a single change point at $\Delta = 150$ and the total length $T = 300$.  If no change point is detected before $T$, then we set $\widehat{\Delta}=\infty$ and  $\tilde{\Delta}=\min\{T, \widehat{\Delta}\}$.

\medskip
\noindent \textbf{Simulation settings.}

\textbf{Scenario 1: SBM.} Let the data be from SBMs with $n = 100$ and three equally-sized communities.  Let $z_i \in \{1, 2, 3\}$ denote the community membership label of node $i$, $i \in \{1, 2, \dots, n\}$. Let the $(i,j)$th entry of the underlying graphon matrix $\Theta(t)$ be $\Theta_{ij}(t) = \vartheta  B_{z_iz_j}(t)$, $i,j \in \{1,2, \dots, n\}$, where $\vartheta  = 0.5$ and $B(t)$'s are
\[
 B(t) = 
\begin{cases}
  \begin{pmatrix}
            0.6 & 1.0 & 0.6\\
            1.0 & 0.6 & 0.5\\
            0.6 & 0.5 & 0.6
            \end{pmatrix}, &  t \in \{1,2, \dots, \Delta\}, \\    
   \begin{pmatrix}
            0.6 & 0.5 & 0.6\\
            0.5 & 0.6 & 1.0\\
            0.6 & 1.0 & 0.6
            \end{pmatrix}, & 
            t \in \{\Delta+1, \dots, T\}. 
\end{cases}
\]
For $t \in \{1, 2, \dots, T\}$, we first generate the unobserved, symmetric adjacency matrices $\{Y(t)\}_{t = 1}^T$ as
\begin{align*}
    Y_{ij}(t) \overset{\mathrm{ind.}}{\sim} \mathrm{Bernoulli}(\Theta_{ij}(t)), \quad  i \leq j,
\end{align*}
followed by the missingness pattern matrices $\{\Omega(t)\}_{t = 1}^T$ as
\begin{align*}
    \Omega_{ij}(t) \overset{\mathrm{ind.}}{\sim} \mathrm{Bernoulli}(\pi), \quad i \leq j,
\end{align*}
where $\pi = q_1 = q_2$ controls the overall proportion of missing entries. 

We let the observed, symmetric adjacency matrices with missing values $\{Y_{\Omega}(t)\}_{t=1}^T$ be
\begin{align*}
    Y_{\Omega,\, ij}(t) = \Omega_{ij}(t) Y_{ij}(t), \quad i \leq j.
\end{align*}

In this scenario, we consider undirected graphs with self-loops, which implies low-rank of the graphon matrices $\{\Theta(t)\}_{t = 1}^T$, i.e.~$r = 3$ across all $t$.

\textbf{Scenario 2: temporally independent RDPG.}   To generate data from RDPGs, we let $n = 100$ and the latent positions be $X$ for $t \in \{1, \ldots, \Delta\}$ and $\widetilde X$ for $\{\Delta+1, \ldots, T\}$. Suppose $X, X^{\prime} \in \mathbb{R}^{n \times 5}$ are generated from $X_{ij}, X^{\prime}_{ij} \overset{\mathrm{i.i.d.}}{\sim} \mathrm{Unif}[0,1]$, $i \in \{1, 2, \dots, n\}$, $j \in \{1, 2, \dots, 5\}$. Then
for $i \in \{1, 2, \dots, n\}$,
\begin{align*}
    \widetilde{X}_i = \begin{cases}
    X^{\prime}_i, & i \leq \lfloor n/4 \rfloor,\\
    X_i, & \mbox{otherwise}.
\end{cases}
\end{align*}
The latent positions $X$ and $\widetilde{X}$ are fixed throughout.
For $i < j$, the $(i,j)$th entries of the unobserved, symmetric adjacency matrices $\{Y_{ij}(t)\}_{t=1}^T$ are generated independently as
\begin{align*}
    Y_{ij}(t) \sim \begin{cases}
      \mathrm{Bernoulli}\bigg(\frac{X_i^{\top}X_j}{\Vert X_i \Vert \Vert X_j \Vert}\bigg), & t \in \{1, \ldots, \Delta\},\\  
      \mathrm{Bernoulli}\bigg(\frac{\widetilde X_i^{\top}\widetilde X_j}{\Vert \widetilde X_i \Vert \Vert \widetilde X_j \Vert}\bigg), & t \in \{\Delta+1,  \ldots, T\}.
    \end{cases} 
\end{align*}
The matrices $\{\Omega(t)\}_{t = 1}^T$ are generated in the same way as in \textbf{Scenario 1} and we fix $\pi = q_1 = q_2 = 0.9$.  Different from \textbf{Scenario 1}, we consider symmetric adjacency matrices without self-loop, a more challenging setting as this does not guarantee a low-rank graphon.  The observed, symmetric adjacency matrices with missing values $\{Y_{\Omega}(t)\}_{t=1}^T$ are generated as $Y_{\Omega, \,ij}(t) = \Omega_{ij}(t) A_{ij}(t)$, $i < j$.

\textbf{Scenario 3: temporally dependent RDPG.}  Let $n = 100$, we generate $n$ latent positions $X_i(1) \sim \mathcal{N}(0, I_5)$, and for $t \in \{2, \dots, \Delta\}$,
\begin{align*}
    X_i(t) \begin{cases}
    = X_i(t-1), & \text{with probability } 0.5,\\
    \overset{\mathrm{ind.}}{\sim} \mathcal{N}(0, I_5), & \text{with probability } 0.5.
\end{cases}
\end{align*}
Let
\begin{align*}
    P_{ij}(t) = \frac{\exp\{X_i(t)^{\top}X_i(t)\}}{1+\exp\{X_i(t)^{\top}X_i(t)\}} \;\; \text{and} \;\; Q_{ij}(t) = \max\left\{q_1, \frac{\exp\{X_i(t)^{\top}X_i(t)\}}{1+\exp\{X_i(t)^{\top}X_i(t)\}}\right\},
\end{align*}
where $q_1 \in \{0.7, 0.8\}$.
After the change point, we generate $n$ latent positions $X_i(\Delta+1) \sim \mathrm{Unif}[0,1]^{\otimes 5}$, and for $t \in \{\Delta+2, \dots, T\}$
\begin{align*}
    X_i(t) \begin{cases}
    = X_i(t-1), & \text{with probability } 0.5,\\
    \overset{\mathrm{ind.}}{\sim} \mathrm{Unif}[0,1]^{\otimes 5}, & \text{with probability } 0.5.
\end{cases}
\end{align*}
Let
\begin{align*}
    P_{ij}(t) = \frac{X_i^{\top}X_j}{\Vert X_i \Vert \Vert X_j \Vert} \;\; \text{and} \;\; Q_{ij}(t) = \max\left\{q_1, \frac{X_i^{\top}X_j}{\Vert X_i \Vert \Vert X_j \Vert}\right\}.
\end{align*}
For $i < j$, we generate independently
\begin{align*}
    Y_{ij}(t)|\{X_i(t), X_j(t)\} \overset{\mathrm{ind.}}{\sim} \mathrm{Bernoulli}\big(P_{ij}(t)\big) \, \mbox{and} \, \Omega_{ij}(t)|\{X_i(t), X_j(t)\} \overset{\mathrm{ind.}}{\sim} \mathrm{Bernoulli}\big(Q_{ij}(t)\big).
\end{align*}

We consider symmetric adjacency matrices without self-loop, which does not guarantee a low-rank graphon. However, the adjacency and missingness matrices are generated from latent positions based on a Markov chain-based RDPG model. As a result, these processes possess temporal dependence.

\medskip
\noindent \textbf{Results}

We collect all the results here. Tables \ref{tab:delay-1} and \ref{tab:delay} show that for the SBM settings, our approach outperforms the competitors by having the smallest average detection delay for varying degrees of missingness, while maintaining exceptional control over proportion of false alarms.  While the performance of RC is close, it fails to keep the desired level of false alarms in some instances.  \Cref{tab: compare} shows that, for the temporally independent RDPG settings, our method outperforms the $k$-NN graph based approach and MissInspect.  The RC approach performs best across the board.  We conjecture that better choices of constants, e.g.~$C_{\lambda}$, in our algorithm may improve our performances further.   In \Cref{tab: compare2}, we collect results from temporally dependent RDPG settings and exclude the $k$-NN based approaches, which perform very poorly. We observe that MI does not perform well in this setting. Although RC provides the shortest detection delay, it does not guarantee the false alarm rate, particularly when the missingness is high (i.e.~$q_1 = 0.7$). Our method performs well in this scenario, as demonstrated in \Cref{tab: compare2}. We conjecture that a blockwise permutation for selecting $\{\epsilon_{s,t}\}$ may further improve our performance, given the temporal dependence in the data.

\iffalse
\begin{table}[t]
\centering
\resizebox{\columnwidth}{!}{%
\begin{tabular}{c|cccccc}
\hline
\multicolumn{1}{l|}{} & \multicolumn{6}{c}{$\{\mathrm{PNM}_t\}_{t = 1}^{300}$}                            \\
$\pi_{\mathrm{LB}}$                    & min    & 1st Qu. & Median & Mean   & 3rd Qu. & Max    \\ \hline
0.7                   & 0.8375 & 0.8458  & 0.8488 & 0.8490 & 0.8525  & 0.8616 \\
0.8                   & 0.8864 & 0.8960  & 0.8986 & 0.8986 & 0.9012  & 0.9093 \\
0.9                   & 0.9428 & 0.9483  & 0.9507 & 0.9503 & 0.9525  & 0.9588 \\ \hline
\end{tabular}
}
\caption{Summary of proportion of non-missing entries by $\pi_{LB}$.}
\label{tab:pnm_t}
\end{table}
\fi

\begin{table}[ht]
\centering
\begin{tabular}{ccccccc}
\hline
 $\alpha$         & SI   & RC  & ORI   & W     & G     & MI  \\ \hline\hline
\multicolumn{7}{c}{$\pi = 0.7$}\\
 0.05          & \textbf{3.12} & 3.82 & 16.09 & 15.06 & 16.05 & 36.38  \\
                 0.01          & \textbf{3.37} & 3.86 & 19.68 & 17.11 & 19.42 & 37.25 \\
\multicolumn{7}{c}{$\pi = 0.8$}\\
 0.05          & \textbf{2.98} & 3.69 & 15.28 & 15.00 & 16.28 & 34.57 \\
              0.01          & \textbf{3.08} & 3.97 & 16.60 & 15.80 & 17.79 & 35.23 \\
\multicolumn{7}{c}{$\pi = 0.9$}\\
 0.05          & \textbf{2.87} & 3.50 & 16.12 & 15.36 & 17.26 & 47.03 \\
              0.01          & \textbf{2.90} & 3.60 & 17.68 & 16.52 & 18.63 & 48.22 \\
\multicolumn{7}{c}{$\pi = 0.95$}\\
 0.05          & \textbf{2.69} & 3.42 & 15.15 & 14.95 & 16.85 & 51.62 \\
              0.01          & \textbf{2.76} & 3.58 & 15.95 & 18.18 & 18.99 & 52.36 \\ \hline
\end{tabular}
\caption{Average detection delays over 100 repetitions of SBM considered in Scenario 1. SI: our proposed methods; RC: rank-constrained maximum likelihood estimator; ORI, W and G: three different statistics based on $k$-NN; MI: MissInspect; $\pi$, the missing probability; $\alpha$, the false alarm controls.}
\label{tab:delay-1} 
\end{table}

\begin{table}[ht]
\centering
\begin{tabular}{ccccccc}
\hline
 $\alpha$         & SI   & RC  & ORI   & W     & G     & MI  \\ \hline\hline
 \multicolumn{7}{c}{$\pi = 0.7$}\\
 0.05          & {0}  & 0.06 & 0.03 & 0.02 & 0.04 & 0.07 \\
                 0.01          & {0}  & 0.04 & {0}    &{ 0}    & {0}    & 0.04 \\
                 \multicolumn{7}{c}{$\pi = 0.8$}\\
 0.05          &  {0}  & 0.05 & 0.04 & 0.03 & 0.06 & {0}    \\
              0.01     & {0}  & 0.01 & 0.02 & 0.02 & 0.01 & {0}    \\
              \multicolumn{7}{c}{$\pi = 0.9$}\\
 0.05          &  {0}  & 0.1  & 0.11 & 0.05 & 0.05 & 0.04 \\
              0.01      & {0}  & 0.04 & 0.01 & 0.03 & 0.01 & 0.02 \\
              \multicolumn{7}{c}{$\pi = 0.95$}\\
 0.05         & {0}  & 0.08 & 0.09 & 0.08 & 0.06 & 0.03 \\
              0.01       & {0}  & 0.05 & 0.06 & 0.02 & 0.03 & 0.01 \\ \hline
\end{tabular}
\caption{False alarms rates over 100 repetitions of SBM considered in Scenario 1. SI: our proposed methods; RC: rank-constrained maximum likelihood estimator; ORI, W and G: three different statistics based on $k$-NN; MI: MissInspect; $\pi$, the missing probability; $\alpha$, the false alarm controls.}
\label{tab:delay} 
\end{table}

\begin{table}[ht]
\centering
\begin{tabular}{ccccccc}
\hline
 $\alpha$         & SI   & RC  & ORI   & W     & G     & MI  \\ \hline\hline
\multicolumn{7}{c}{Delay}\\
 0.05         & 9.67 & \textbf{4.00} & 23.53 & 24.26 & 24.85 & 19.53  \\
 0.01         & 9.62 & \textbf{4.04} & 26.58 & 25.85 & 26.88 & 21.60 \\
\multicolumn{7}{c}{PFA}\\
 0.05          & {0}  & {0}    & 0.04 & 0.04 & 0.01 & 0.11 \\
              0.01          & {0}  & {0}    & {0}    & 0.01 & 0.01 & {0}\\
\hline
\end{tabular}
\caption{Average detection delays and false alarms rates over 100 repetitions of RDPG considered in Scenario 2, with the missing probability $\pi = q_1 = q_2 =0.9$. SI: our proposed methods; RC: rank-constrained maximum likelihood estimator; ORI, W and G: three different statistics based on $k$-NN; MI: MissInspect; $\alpha$, the false alarm controls. }
\label{tab: compare}
\end{table}

\begin{table}  
\centering
\begin{tabular}{lllllll}
    \toprule
    $\alpha$ & SI & RC & MI & SI & RC & MI \\
    \midrule
    & \multicolumn{3}{c}{Delay,\; $q_1 = 0.7$} & \multicolumn{3}{c}{Delay,\; $q_1 = 0.8$}\\
    \cmidrule(r){2-4} \cmidrule(r){5-7} 
    0.05 & 5.17 & 1.88 & 23.07 & 4.60 & 1.91 & 22.33 \\
    0.01 & 5.47 & 2.04 & 25.48 & 5.10 & 2.34 & 23.52 \\ [3pt]
    & \multicolumn{3}{c}{PFA,\; $q_1 = 0.7$} & \multicolumn{3}{c}{PFA,\; $q_1 = 0.8$} \\
    \cmidrule(r){2-4} \cmidrule(r){5-7} 
    0.05 & 0.02 & 0.15 & 0.28 & 0.06 & 0.08 & 0.20 \\
    0.01 & 0 & 0.07 & 0.08 & 0 & 0 & 0.10\\
    \bottomrule
  \end{tabular}
  \caption{Average detection delays and false alarms rates over 100 repetitions of RDPG considered in Scenario 3, with the missingness probabilities $q_1 \in \{0.7, 0.8\}$ and $q_2 = 1$. SI: our proposed methods; RC: rank-constrained maximum likelihood estimator; MI: MissInspect; $\alpha$, the false alarm controls. }
  \label{tab: compare2}
\end{table}

\subsection{Real data analysis} \label{sec-data}  

We apply the proposed online change point detection method and competitors discussed in \Cref{sec-simulation} to networks obtained from the weekly log-returns of $29$ companies, based on the Dow Jones Industrial Average index from April 1990 to January 2012 \citep[see the R package \texttt{ecp},][]{james2013ecp}.  The data are processed as follows.  We first obtain sequences of outer-product matrices of the companies’ weekly log-returns at each week and then threshold the entries of matrices by setting the entries with values above the $95\%$ quantile as $1$ and zero otherwise.  Following this we apply the missing scheme described in \textbf{Scenario 1} with $\pi= 0.9$. Finally, the network sequence, with missing values, together with their associated missingness pattern matrices are treated as the observed data.  We consider two time periods: period 1 is from 2-Apr-1990 to 31-May-2004 and period 2 is from 31-May-2004 to 1-Mar-2010.  

\begin{table}[ht]
\centering
\begin{tabular}{cccc}
\hline
        $\alpha$ & SI         & RC & ORI     \\ \hline\hline
        \multicolumn{4}{c}{Period 1} \\
 0.05  & 2001-10-22 & Inf  & Inf \\
0.01  & 2001-12-17 & Inf  & Inf  \\
\multicolumn{4}{c}{Period 2} \\
 0.05  & 2008-04-14 & 2008-06-23  & 2009-08-03 \\
 0.01  & 2008-06-16 & 2009-06-01  & 2009-08-03 \\ \hline 
$\alpha$ & W        & G & MI    \\ \hline\hline
        \multicolumn{4}{c}{Period 1} \\
 0.05  & 2002-12-30 & Inf  & 2001-07-23 \\
0.01  & Inf & Inf & Inf \\
\multicolumn{4}{c}{Period 2} \\
 0.05  & 2009-04-20 & 2009-08-03 & 2008-12-22 \\
 0.01  &2009-08-03 & 2009-08-03 & Inf \\ \hline 
\end{tabular}
\caption{Stock exchange networks: Estimated change points with the missing probability $\pi=0.9$.  SI: our proposed methods; RC: rank-constrained maximum likelihood estimator; ORI, W and G: three different statistics based on $k$-NN; MI: MissInspect; $\alpha$, the false alarm controls.}
\label{tab:data}
\end{table}

For period 1, we let the data dated from 2-Apr-1990 to 4-Jan-1999 be the training set, used to calibrate the tuning parameters, and detect change points in the rest in an online fashion.  To be specific, instead of letting all data available at once, we pretend that the data are obtained sequentially.  Note that there was some dramatic financial turbulence right after the terrorist attacks on 11-Sept-2001.  Our proposed method successfully detects a change point at 22-Oct-2001 (17-Dec-2001) with $\alpha = 0.05$ (0.01), indicating a small detection delay.  As for the competitors, the weighted edge count statistic (W) detects a much later change point and MissInspect (MI) detects the date 23-July-2001, which appears more likely to be a false alarm. 

For period 2, we let the data dated from 31-May-2004 to 15-Jan-2007 be the training set and detect change points in the rest in an online fashion.  Note that there was a financial crisis in 2007-2008.  We detect a change point in April (June) 2008 with $\alpha = 0.05$ (0.01), with a significantly smaller detection delay than the competitors.

\subsection{Additional simulation results}\label{sec-addition}

We now demonstrate the effect of the low-rank assumption and the robustness of the choices of $C_{\lambda}$.

Let the network size $n = 100$ and let $b$ be the number of equally-sized communities, each of which has size around $n/b$.  Let $z_i \in \{1, \ldots, b\}$ be the community membership of node $i$, $i \in \{1, \ldots, n\}$.  Let the $(i,j)$th entry of the underlying graphon matrix be $\Theta_{ij}(t) = \vartheta  B_{z_iz_j}(t)$, where $\vartheta  = 0.5$ and the matrices $B(t)$'s are
\begin{align*}
    B(t) = \begin{cases}
  0.2 \times \mathbf{1}_{b} + 0.7 \times \mathbf{I}_b, & t \in \{1,2, \dots, \Delta\}, \\    
   0.1 \times \mathbf{1}_b + 0.4 \times \mathbf{I}_b, & t \in \{\Delta+1, \Delta+2, \dots, T\}, 
\end{cases}
\end{align*}
where $\mathbf{1}_b$ is a $b \times b$ matrix of ones and $\mathbf{I}_b$ is an identity matrix of size $b$.
For $t \in \{1, 2, \dots, T\}$, we first generate the unobserved, symmetric adjacency matrices $\{Y(t)\}_{t = 1}^T$ as
\begin{align*}
    Y_{ij}(t) \overset{\mathrm{ind.}}{\sim} \mathrm{Bernoulli}(\Theta_{ij}(t)), \quad  i \leq j,
\end{align*}
followed by the missingness matrices $\{\Omega(t)\}_{t = 1}^T$ as
\begin{align*}
    \Omega_{ij}(t) \overset{\mathrm{ind.}}{\sim} \mathrm{Bernoulli}(\pi), \quad i \leq j,
\end{align*}
where $\pi = q_1 = q_2 = 0.9$.   Note that we consider symmetric adjacency matrices with self-loop, which ensures the symmetrized graphon matrix of rank $b$. We consider $b \in \{3, 5, 10, 25, 50\}$. 

\textbf{Varying rank.}  Fixing $C_{\lambda} = 0.67$, the performance of our method is reported in \Cref{tab:delay-s1}, implying the trend that the detection delay increases  as the rank $b$ increases.  This is as expected and is indicated in \Cref{thm-upper-bound}.  
\begin{table}[ht]
\centering
\begin{tabular}{cccccc}
\hline
 $\alpha$         & $b = 3$   & $b = 5$  & $b = 10$   & $b = 25$     & $b = 50$   \\ \hline\hline
\multicolumn{6}{c}{Delay}\\
 0.05          & 3.33 & 3.90 & 6.67 & 8.81 & 8.50  \\
                 0.01          & 3.62 & 4.29 & 6.78 & 8.21 & 8.33 \\
\multicolumn{6}{c}{PFA}\\
 0.05          & 0 & 0 & 0 & 0 & 0  \\
              0.01          & 0 & 0 & 0 & 0 & 0  \\ \hline
\end{tabular}
\caption{Average detection delays and false alarms
rates of our proposed method over 100 repetitions of SBM, with $q_1 = q_2 = 0.9$. $\alpha$, the false alarm controls; b, the number of communities of underlying graphon.}
\label{tab:delay-s1} 
\end{table}

\textbf{Robustness of the choices of $C_{\lambda}$.}  We let $C_{\lambda} \in \{0.25, 0.50, 0.67, 1.00, 1.25\}$ and collect results in \Cref{tab:delay-s2}.  It can be seen that the $C_{\lambda}$, which produces the smallest detection delay among candidates, varies in different settings. % Although the performance of our method is robust to the considered $C_{\lambda}$'s, this suggests a data-driven choice of $C_{\lambda}$, we will investigate this in the future.
\begin{table}[ht]
\centering
\begin{tabular}{cccccc}
\hline
 $\alpha$         & $C_{\lambda} = 0.25$   & $C_{\lambda} = 0.50$  & $C_{\lambda} = 0.67$   & $C_\lambda = 1.00$     & $C_{\lambda} = 1.25$   \\ \hline\hline
\multicolumn{6}{c}{b = 3}\\
 0.05          & 6.15 (0.02) & \textbf{2.72} (0.02) & 3.33 (0) & 3.21 (0) & 3.76 (0)  \\
                 0.01          & 6.15 (0.01) & \textbf{2.97} (0) & 3.62 (0) & 3.43 (0) & 3.85 (0) \\
\multicolumn{6}{c}{b = 5}\\
 0.05          & 6.64 (0) & 4.10 (0) & \textbf{3.90} (0) & 4.95 (0) & 5.98 (0)  \\
              0.01          & 6.36 (0) & 4.57 (0) & \textbf{4.29} (0) & 5.29 (0) & 6.31 (0) \\
\multicolumn{6}{c}{b = 10}\\
 0.05          & 8.38 (0) & \textbf{6.01} (0) & 6.67 (0) & 9.68 (0) & 13.29 (0)  \\
              0.01          & 8.00 (0) & \textbf{5.91} (0) & 6.78 (0) & 9.67 (0) & 13.27 (0) \\
\multicolumn{6}{c}{b = 25}\\
 0.05          & 15.90 (0) & 9.78 (0) & \textbf{8.81} (0) & 9.18 (0) & 9.25 (0)  \\
              0.01          & 14.60 (0) & 9.27 (0) & \textbf{8.21} (0) & 8.79 (0) & 9.01 (0) \\
\multicolumn{6}{c}{b = 50}\\
 0.05          & 16.89 (0) & 10.73 (0) & 8.50 (0) & 5.28 (0) & \textbf{4.60} (0)  \\
              0.01          & 15.57 (0) & 10.65 (0) & 8.33 (0) & 5.08 (0) & \textbf{4.52} (0) \\
              \hline
\end{tabular}
\caption{Average detection delays and false alarms
rates (in parenthesis) of our proposed method over 100 repetitions of SBM, with $q_1 = q_2 = 0.9$. $\alpha$, the false alarm controls; b, the number of communities.}
\label{tab:delay-s2} 
\end{table}

\section{Conclusion}

In this paper, we study an online change point detection in a dynamic networks with heterogeneous missingness.  A highlight of our framework is that we allow for dependence: edge dependence and missingness dependence within a network, temporal dependence across time and the dependence between the edge pattern and missingness pattern.  Despite the technical condition on the $\phi$-mixing coefficients, we have demonstrated its practicality through concrete examples.  Having such a general framework, we are still able to show a near-optimal detection delay upper bound, subject to an overall Type-I error control.

\bibliography{ref}
\bibliographystyle{apalike}

\appendix

\section*{Appendices}
%More discussions on the theoretical results derived, additional numerical experiments results and all technical details are collected in this document.

\section[]{Proof of results in \Cref{sec-problem-setup}}

\begin{proof}[Proof of \Cref{theorem:markov-phi}]
    It suffice to consider a fixed $(i,j) \in \{1, \ldots, n\}^{\otimes 2}$. 
    Let $\{U_{ij}(t)\}_{t \in \mathbb{N}^*}$ be a sequence of i.i.d. $\textrm{Unif}~[0,1]$ random variables, which are independent of $\{X_i(t), X_j(t)\}_{t \in \mathbb{N}^*}$. We have equivalently that
    \begin{equation*}
        Y_{ij}(t) = \mathbbm{1}\{X_{i}(t)^{\top}X_j(t) + U_{ij}(t) \geq 1\}.
    \end{equation*}
    Since $\mathbbm{1}\{\cdot \geq 1\}: [0,2] \mapsto \{0, 1\}$ is a measurable function, for any $t \in \mathbb{N}^*$, we have that the $\sigma$-algebra satisfies
    \begin{align*}
        \sigma\big(\mathbbm{1}\{X_{i}(t)^{\top}X_j(t) + U_{ij}(t) \geq 1\}\big)
        \subset \sigma\big(X_{i}(t),X_{j}(t), U_{ij}(t)\big).
    \end{align*}
    Due to the temporal independence, the $\phi$-mixing coefficient of $\{U_{ij}(t)\}_{t \in \mathbb{N}^*}$ is zero for any $\ell \geq 1$. By \Cref{lemma:mixing_indep_comp}, we have that for any $\ell \in \mathbb{N}$,
    \begin{equation*}
        \phi_{Y_{ij}}(\ell) \leq 2\max_{i = 1}^n\phi_{X_i}(\ell) \leq 2C\rho^{\ell},
    \end{equation*}
    which completes the proof.
\end{proof}

\section[]{Proofs of results in \Cref{sec-main}}

\subsection[]{Proof of \Cref{lem-soft-impute-property}}

\begin{proof}[Proof of \Cref{lem-soft-impute-property}]
The results have two parts and we prove them separately.  The first part is an adaptation of the proof Lemma~1 in \cite{mazumder2010spectral} and the second part is an adaptation of the proof Lemma~1 in \cite{klopp2015matrix}.

\medskip
\noindent \textbf{The form of $\widehat{Z}$.}

For any matrix $Z \in \mathbb{R}^{n \times n}$, let $Z = \widetilde{U} \widetilde{D} \widetilde{V}^{\top}$ be a singular value decomposition of $Z$.  Note that
	\begin{align}
		& \frac{1}{2m} \sum_{t = 1}^{m} \|R(t) - Z\|_{\mathrm{F}}^2 + \lambda\|Z\|_* = \frac{1}{2m} \sum_{t = 1}^{m} \left\{\left\|R(t)\right\|_{\mathrm{F}}^2 - 2\sum_{i = 1}^n \tilde{d}_i \tilde{u}_i^{\top} R(t) \tilde{v}_i + \sum_{i = 1}^n \tilde{d}_i^2\right\} + \lambda \sum_{i = 1}^n \tilde{d}_i \nonumber \\
		= & \frac{1}{2} \left\{\frac{1}{m}\sum_{t = 1}^{m}\|R(t)\|_{\mathrm{F}}^2 - 2 \sum_{i = 1}^n \tilde{d}_i \tilde{u}_i^{\top} \left[\frac{1}{m} \sum_{t = 1}^{m} R(t)\right] \tilde{v}_i + \sum_{i = 1}^n \tilde{d}_i^2\right\} + \lambda \sum_{i = 1}^n \tilde{d}_i, \label{eq-proof-lem1-1}
	\end{align} 	
	where 
    \[
		\widetilde{D} = \mathrm{diag}\{\tilde{d}_1, \ldots, \tilde{d}_n\}, \quad \widetilde{U} = (\tilde{u}_1, \ldots, \tilde{u}_n) \quad \mbox{and} \quad \widetilde{V} = (\tilde{v}_1, \ldots, \tilde{v}_n),
	\]
    and the first identity follows from
    \[
        \mathrm{tr}\{Z^{\top}R(t)\} = \mathrm{tr}\{\widetilde{V} \widetilde{D} \widetilde{U}^{\top} R(t)\} = \sum_{i = 1}^n \mathrm{tr} \{\tilde{d}_i \tilde{v}_i \tilde{u}_i^{\top} R(t)\} = \sum_{i = 1}^n \tilde{d}_i \tilde{u}_i^{\top} R(t)\tilde{v}_i.
    \]	
	
Minimising \eqref{eq-proof-lem1-1} is equivalent to minimising
	\[
		- 2 \sum_{i = 1}^n \tilde{d}_i \tilde{u}_i^{\top} \left[\frac{1}{m} \sum_{t = 1}^{m} R(t)\right] \tilde{v}_i + \sum_{i = 1}^n \tilde{d}_i^2 + 2\lambda \sum_{i = 1}^n \tilde{d}_i, 
	\]
	with respect to $(\tilde{u}_i, \tilde{v}_i, \tilde{d}_i)$, $i = 1, \ldots, n$, under the constraints $\widetilde{U}^{\top}\widetilde{U} = \widetilde{V}^{\top}\widetilde{V} = I_n$ and $\tilde{d}_i \geq 0$.
	
Observe that the above is equivalent to minimising (with respect to $\widetilde{U}$ and $\widetilde{V}$) the function $g(\widetilde{U}, \widetilde{V})$,
	\begin{equation}\label{eq-proof-lem1-2}
		g(\widetilde{U}, \widetilde{V}) = \min_{\widetilde{D} \geq 0} \left\{- 2 \sum_{i = 1}^n \tilde{d}_i \tilde{u}_i^{\top} \left[\frac{1}{m} \sum_{t = 1}^{m} R(t)\right] \tilde{v}_i + \sum_{i = 1}^n \tilde{d}_i^2 + 2\lambda \sum_{i = 1}^n \tilde{d}_i\right\}.
	\end{equation}
	Since the objective \eqref{eq-proof-lem1-2} to be minimised with respect to $\widetilde{D}$, is separable in $\tilde{d}_i$, $i = 1, \ldots, n$.  It suffices to minimise it with respect to each $\tilde{d}_i$ separately.
	
The problem 
	\[
		\min_{\tilde{d}_i \geq 0} \left\{- 2 \tilde{d}_i \tilde{u}_i^{\top} \left[\frac{1}{m} \sum_{t = 1}^{m} R(t)\right] \tilde{v}_i + \tilde{d}_i^2 + 2\lambda \tilde{d}_i\right\}
	\]	
	can be solved by the solution 
	\[
		S_{\lambda}\left(\tilde{u}_i^{\top} \left[\frac{1}{m} \sum_{t = 1}^{m} R(t)\right] \tilde{v}_i\right) = \left(\tilde{u}_i^{\top} \left[\frac{1}{m} \sum_{t = 1}^{m} R(t)\right] \tilde{v}_i - \lambda\right)_+.
	\]
	Plugging this into \eqref{eq-proof-lem1-2} leads to 
	\begin{align*}
		g(\widetilde{U}, \widetilde{V}) & = \frac{1}{2} \Bigg\{\frac{1}{m}\sum_{t = 1}^{m} \|R(t)\|_{\mathrm{F}}^2 -  2\sum_{i = 1}^n \left(\tilde{u}_i^{\top} \left[\frac{1}{m}\sum_{t = 1}^{m} R(t)\right] \tilde{v}_i - \lambda\right)_+\left(\tilde{u}_i^{\top} \left[\frac{1}{m}\sum_{t = 1}^{m} R(t)\right] \tilde{v}_i - \lambda\right) \\
		& \hspace{1cm} + \sum_{i = 1}^n \left(\tilde{u}_i^{\top} \left[\frac{1}{m}\sum_{t = 1}^{m} R(t)\right] \tilde{v}_i - \lambda\right)_+^2\Bigg\}.
	\end{align*}
	Minimising $g(\widetilde{U}, \widetilde{V})$ with respect to $(\widetilde{U}, \widetilde{V})$ is equivalent to maximising
	\begin{equation}\label{eq-proof-lemma1-singular-values}
	    \sum_{\tilde{u}_i^{\top} \left[\frac{1}{m}\sum_{t = 1}^{m} R(t)\right] \tilde{v}_i > \lambda} \left(\tilde{u}_i^{\top} \left[\frac{1}{m}\sum_{t = 1}^{m} R(t)\right] \tilde{v}_i - \lambda\right)_+^2. 
	\end{equation}
    Due to the definitions of singular values, we see that \eqref{eq-proof-lemma1-singular-values} is solved by the left and right singular vectors of the matrix $m^{-1}\sum_{t = 1}^{m} R(t)$.  The final result then follows.

\medskip

\noindent \textbf{The convergence.}

Since for any $t \in \{1, \ldots, m\}$ and $k \in \mathbb{N}$, it holds that 
	\[
		\|(M_{k+1} - M_k)_{\overline{\Omega(t)}}\|_{\mathrm{op}} \leq \|(M_{k+1} - M_k)_{\overline{\Omega(t)}}\|_{\mathrm{F}} \leq \|M_{k+1} - M_k\|_{\mathrm{F}}
	\]	
	and
	\[
		\|M_{k+1} - M_k\|_{\infty} \leq \|M_{k+1} - M_k\|_{\mathrm{F}},
	\]
	it suffices to show that $\|M_{k+1} - M_k\|_{\mathrm{F}} \to 0$, as $k \to \infty$.

Denote by $M_k^*$ the $k$th step solution before truncating.
For any $k \in \mathbb{N}^*$, we have that 
	\begin{align}
		& \|M_{k+1} - M_k\|_{\mathrm{F}} \leq \|M^*_{k+1} - M^*_k\|_{\mathrm{F}} \nonumber \\
		= & \left\|S_{\lambda} \left\{m^{-1} \sum_{t = 1}^{m} \left[(Y(t))_{\Omega(t)} + (M_k)_{\overline{\Omega(t)}}\right]\right\} - S_{\lambda} \left\{m^{-1} \sum_{t = 1}^{m} \left[(Y(t))_{\Omega(t)} + (M_{k-1})_{\overline{\Omega(t)}}\right]\right\}\right\|_{\mathrm{F}} \nonumber \\
		\leq & \left\|\left\{m^{-1} \sum_{t = 1}^{m} \left[(Y(t))_{\Omega(t)} + (M_k)_{\overline{\Omega(t)}}\right]\right\} - \left\{m^{-1} \sum_{t = 1}^{m} \left[(Y(t))_{\Omega(t)} + (M_{k-1})_{\overline{\Omega(t)}}\right]\right\}\right\|_{\mathrm{F}} \nonumber \\
		= & \left\|m^{-1} \sum_{t = 1}^{m} (M_k - M_{k-1})_{\overline{\Omega(t)}}\right\|_{\mathrm{F}} \leq \|M_k - M_{k-1}\|_{\mathrm{F}}. \label{eq-pf-lem-1-klopp-1}
	\end{align}
	This implies that the sequence $\{\|M_k - M_{k-1}\|_{\mathrm{F}}\}_{k \in \mathbb{N}^*}$ converges.  It remains to show that this sequence converges to zero.
	
Note that \eqref{eq-pf-lem-1-klopp-1} imply that 
	\[
		\|M_k - M_{k-1}\|_{\mathrm{F}}^2 - \left\|m^{-1} \sum_{t = 1}^{m} (M_k - M_{k-1})_{\overline{\Omega(t)}}\right\|_{\mathrm{F}}^2 = \left\|m^{-1} \sum_{t = 1}^{m} (M_k - M_{k-1})_{\Omega(t)}\right\|_{\mathrm{F}}^2 \to 0,
	\]	  
	as $k \to \infty$.  We therefore only need to show that 
	\begin{equation}\label{eq-lemma1-proof-sufficient-cond}
	    \left\|m^{-1} \sum_{t = 1}^{m} (M_k - M_{k-1})_{\overline{\Omega(t)}}\right\|_{\mathrm{F}} \to 0, \quad \mbox{as } k \to \infty.
	\end{equation}
	
Define
	\[
		h(A, B) = \frac{1}{2m} \sum_{t = 1}^{m} \left\|(Y(t) - B)_{\Omega(t)}\right\|_{\mathrm{F}}^2 + \frac{1}{2m} \sum_{t = 1}^{m} \left\|(A - B)_{\overline{\Omega(t)}}\right\|_{\mathrm{F}}^2 + \lambda \|B\|_*.
	\]
	Due to the definition of $M^*_{k+1}$, we have that
	\begin{align*}
		& h(M_k, M^*_k) = \frac{1}{2m} \sum_{t = 1}^{m} \left\|(Y(t) - M^*_k)_{\Omega(t)}\right\|_{\mathrm{F}}^2 + \frac{1}{2m} \sum_{t = 1}^{m} \left\|(M_k - M^*_k)_{\overline{\Omega(t)}}\right\|_{\mathrm{F}}^2 + \lambda \|M^*_k\|_* \\
        \geq & \frac{1}{2m} \sum_{t = 1}^{m} \left\|(Y(t) - M^*_{k+1})_{\Omega(t)}\right\|_{\mathrm{F}}^2 + \frac{1}{2m} \sum_{t = 1}^{m} \left\|(M_k - M^*_{k+1})_{\overline{\Omega(t)}}\right\|_{\mathrm{F}}^2 + \lambda \|M^*_{k+1}\|_* \\
        = & h(M_k, M^*_{k+1}) \geq \frac{1}{2m} \sum_{t = 1}^{m} \left\|(Y(t) - M^*_{k+1})_{\Omega(t)}\right\|_{\mathrm{F}}^2 + \frac{1}{2m}\sum_{t = 1}^{m}\left\|(M_{k+1} - M^*_{k+1})_{\overline{\Omega(t)}}\right\|_{\mathrm{F}}^2 + \lambda \|M^*_{k+1}\|_* \\
		= & h(M_{k+1}, M^*_{k+1}),
	\end{align*}
	which shows that the sequence $\{h(M_k, M^*_k)\}_{k \geq 1}$ converges.  Therefore,
	\begin{align*}
		& h(M_k, M^*_{k+1}) - h(M_{k+1}, M^*_{k+1}) \\
        = & \frac{1}{2m}\sum_{t = 1}^{m}\left\|(M_k - M^*_{k+1})_{\overline{\Omega(t)}}\right\|_{\mathrm{F}}^2 - \frac{1}{2m}\sum_{t = 1}^{m}\left\|(M_{k+1} - M^*_{k+1})_{\overline{\Omega(t)}}\right\|_{\mathrm{F}}^2 \to 0.
	\end{align*}

Since
	\begin{align*}
		& \frac{1}{2m}\sum_{t = 1}^{m}\left\|(M_k - M^*_{k+1})_{\overline{\Omega(t)}}\right\|_{\mathrm{F}}^2 - \frac{1}{2m}\sum_{t = 1}^{m}\left\|(M_{k+1} - M^*_{k+1})_{\overline{\Omega(t)}}\right\|_{\mathrm{F}}^2 \geq \frac{1}{2m}\sum_{t = 1}^{m} \left\|(M_k - M_{k+1})_{\overline{\Omega(t)}}\right\|_{\mathrm{F}}^2,
	\end{align*}
    combining with \eqref{eq-lemma1-proof-sufficient-cond}, we complete the proof.
\end{proof}

\subsection[]{Proofs of \Cref{thm-upper-bound} and \Cref{cor-main-klopp-cor-3}}

\subsubsection{Notation}

For notational simplicity, in the proof, we adopt an equivalence notation system.  Let the observations be $\{Y_{\Omega}(t), \Omega(t)\}_{t \in \mathbb{N}^*}$, and for each $(i, j) \in \{1, \ldots, n\}^{\otimes 2}$
	\[
		Y_{\Omega, ij}(t) = \Omega_{ij}(t)\{P_{ij}(t) + \xi_{ij}(t)\},
	\]
where $P_{ij}(t) = \{X_i(t)\}^{\top}\{X_j(t)\}$ and the noise $\xi_{ij}(t)$ satisfies that 
	\begin{align*}
		\mathbb{P}\big(\xi_{ij}(t) = 1 - P_{ij}(t)|\{X_i(t), X_j(t)\}\big) = 1 - \mathbb{P}\big(\xi_{ij}(t) = -P_{ij}(t)|\{X_i(t), X_j(t)\}\big) = P_{ij}(t).
	\end{align*}
Let $E_{ij} = e_i e_j^{\top}$ be the matrix basis and 
	\begin{equation} \label{eq-Sigma-mat-def}
	    \Sigma(t) = \sum_{(i, j)} \left(\Omega_{ij}(t)\big\{P_{ij}(t) + \xi_{ij}(t) - \Theta_{ij}(t)\big\}\right) E_{ij}.
	\end{equation}
and
    \begin{equation} \label{eq-Sigma-prime-maat-def}
        \Sigma_{R}(t) = \sum_{(i,j)}\big\{\Omega_{ij}(t) - \Pi_{ij}(t)\big\}E_{ij}.
    \end{equation}
Note that under Assumption~\ref{assume-model-new}, we have the sequence $\{\Omega(t)\}_{t \in \mathbb N^*}$ is independent of $\{\xi(t)\}_{t \in \mathbb N^*}$ and $\{P(t)\}_{t \in \mathbb N^*}$. Thus, we have that
\[
\mathbb E[\Sigma(t)] = \mathbb E\left[\mathbb E[\Sigma(t)|\{X_i(t)\}_{i = 1}^n]\right] = 0_{n \times n}.
\]

For any linear vector subspace $S \subset \mathbb{R}^n$, let $P_S$ be the projector onto $S$ and $S^{\perp}$ be the orthogonal complement of $S$.  For any matrix $A$, let $\{u_j(A)\}$ and $\{v_j(A)\}$ be the left and right orthonormal singular vectors of $A$, respectively.  Let $S_1(A) = \mathrm{span}\{u_j(A)\}$ and $S_2(A) = \mathrm{span}\{v_j(A)\}$.  For any matrix $B$, let
    \begin{equation}\label{eq-mathbfP-def}
        \mathbf{P}_A(B)^{\perp} = P_{S_1^{\perp}(A)} BP_{S_2^{\perp}(A)} \quad \mbox{and} \quad \mathbf{P}_A(B) = B - \mathbf{P}_A(B)^{\perp}.
    \end{equation}
For any matrix $A \in \mathbb{R}^{n \times n}$, define the weighted Frobenius norm of $A$ as
    \begin{equation}\label{eq-weightedFrob-def}
        \|A\|^2_{L_2(U)} = \sum_{(i, j)} U_{ij} A_{ij}^2,
    \end{equation}
    where the weights are given by a determinestic matrix $U = (U_{ij})_{(i,j) \in \{1, \ldots, n\}^{\otimes 2}}$.
	
\subsubsection{Estimation bounds}	
	
We collect results on the estimation error bounds on the soft-impute estimators in this subsection.  The results here are based on their counterparts in \cite{klopp2015matrix}.  Note that \cite{klopp2015matrix} is only concerned with one matrix observation, while we deal with a multiple-copy version and allow for temporal and cross-sectional dependence.

\iffalse
{\color{blue}
\begin{proposition}
Assume that Assumptions~\ref{assume-model-new} and \ref{assume-no-change} hold.  Denote $\Theta(t) = \Theta_0$, $t \in \mathbb{N}^*$.  For any integer pair $0 \leq s < e$, define $\Sigma_{s:e} = (e-s)^{-1} \sum_{t = s+1}^e \Sigma(t)$, where $\Sigma(t)$ is defined in \eqref{eq-Sigma-mat-def}.  Let $\widehat{\Theta}_{s:e}$ be the output of \Cref{alg-main-klopp} with 
    \begin{equation}\label{eq-prop-3-lambda-cond}
        \lambda \geq 3\|\Sigma_{s:e}\|_{\mathrm{op}},        
    \end{equation}
    initialiser $\widetilde{M}$ and \textcolor{red}{$a = \rho$}.

With probability at least $1 - \delta$, it holds that 
	\begin{align*}
	    \|\widehat{M}_{s:e} - M_0\|_{\mathrm{F}}^2 \leq \frac{C}{p^2} \Bigg\{\mathrm{rank}(M_0)\left[\lambda^2 + \rho^2 \left\{\mathbb{E}(\|\Sigma_{R, s:e}\|_{\mathrm{op}})\right\}^2\right]  + \frac{\rho^2}{e-s} + \frac{\log(1/\delta)}{e-s}\Bigg\},
	\end{align*}
	where $\Sigma_{R, s:e} = (e-s)^{-1} \sum_{t = s+1}^e \Sigma_R(t)$ and $C > 0$ is an absolute constant.
\end{proposition}
}
\fi

\begin{proposition} \label{thm-main-klopp-cor-3}
Assume that Assumptions~\ref{assume-model-new} and \ref{assume-no-change} hold.  Denote $\Theta(t) = \Theta_0$ for $t \in \mathbb{N}^*$.  For any integer pair $0 \leq s < e$, let $\Sigma_{s:e} = (e-s)^{-1} \sum_{t = s+1}^e \Sigma(t)$, where $\Sigma(t)$ is defined in \eqref{eq-Sigma-mat-def}.  Let $\widehat{\Theta}_{s:e}$ be the output of \Cref{alg-main-klopp} with 
    \begin{equation}\label{eq-prop-3-lambda-cond}
        \lambda_{s,e} \geq 3\|\Sigma_{s:e}\|_{\mathrm{op}},        
    \end{equation}
    $a = \vartheta $ given in Assumption~\ref{assume-model-new} and a generic initialiser $\widehat{\Theta} \in \{0,1\}^{n \times n}$. Then, with probability at least $1 - \delta$, it holds that 
	\begin{align*}
	    \|\widehat{\Theta}_{s:e} - \Theta_0\|_{\mathrm{F}}^2 \leq \frac{C}{q_1^2} \Bigg\{\mathrm{rank}(\Theta_0)\left(\lambda_{s,e}^2 + \vartheta^2 \left\{\mathbb{E}\left[\|\Sigma_{R,s:e}\|_{\mathrm{op}}\right]\right\}^2\right) + \frac{\log(1/\delta)}{e-s}\Bigg\},
	\end{align*}
	where $\Sigma_{R,s:e} = (e-s)^{-1}\sum_{t = s+1}^e\Sigma_R(t)$ and $\Sigma_R(t)$ is defined in \eqref{eq-Sigma-prime-maat-def}, and $C > 0$ is an absolute constant.
\end{proposition}

\begin{proof}[Proof of \Cref{thm-main-klopp-cor-3}]
Since $\|\widehat{\Theta}_{s:e} - \Theta_0\|_{\mathrm{F}}^2 \leq \|\widetilde{\Theta}_{s:e} - \Theta_0\|_{\mathrm{F}}^2$, it suffice to show that
\begin{align*}
	    \|\widetilde{\Theta}_{s:e} - \Theta_0\|_{\mathrm{F}}^2 \leq \frac{C}{q_1^2} \Bigg\{\mathrm{rank}(\Theta_0)\left(\lambda_{s,e}^2 + \vartheta^2 \left\{\mathbb{E}\left[\|\Sigma_{R,s:e}\|_{\mathrm{op}}\right]\right\}^2\right) + \frac{\log(1/\delta)}{e-s}\Bigg\}.
	\end{align*}
Recall that 
    \begin{align*}
        \widetilde{\Theta}_{s:e} & \in \argmin_{\Theta \in \mathbb{R}^{n \times n}} f_{\lambda_{s,e}, s:e} = \argmin_{\Theta \in \mathbb{R}^{n \times n}} f(\Theta; \{Y_{\Omega}(t), \Omega(t)\}_{t = s+1}^e, \widehat{\Theta}, \lambda_{s,e}) \\
        & = \argmin_{\Theta \in \mathbb{R}^{n \times n}}  \left\{\frac{1}{2(e-s)} \sum_{t = s+1}^e \left\|Y_{\Omega}(t) + \widehat{\Theta}_{\overline{\Omega(t)}} - \Theta\right\|_{\mathrm{F}}^2 + \lambda_{s,e}\|\Theta\|_*\right\}.
    \end{align*}
    
\medskip
\noindent \textbf{Step 1.} It follows from \Cref{lem-soft-impute-property} that $\widetilde{\Theta}_{s:e}$ minimises $f_{\lambda_{s,e}, s:e}(\Theta)$.  Using the sub-gradient stationary conditions we have that
	\begin{equation}\label{eq-proof-prop3-step1-1}
	    - \Bigg\langle \frac{1}{e-s} \sum_{t = s+1}^e \left\{Y_{\Omega}(t) + \widehat{\Theta}_{\overline{\Omega(t)}}\right\} - \widetilde{\Theta}_{s:e}, \, \widetilde{\Theta}_{s:e} - \Theta_0\Bigg\rangle + \lambda_{s,e} \left\langle \widehat{V}, \, \widetilde{\Theta}_{s:e} - \Theta_0\right\rangle \leq 0,
	\end{equation}
	where $\widehat{V} \in \partial \|\widetilde{\Theta}_{s:e}\|_*$.  We have that
	\begin{align}
		& \sum_{t = s+1}^e \left\langle \left(Y(t) - \Theta_0\right)_{\Omega(t)}, \, \widetilde{\Theta}_{s:e} - \Theta_0\right\rangle + \sum_{t = s+1}^e \left\langle \left(\widehat{\Theta} - \widetilde{\Theta}_{s:e}\right)_{\overline{\Omega(t)}}, \, \widetilde{\Theta}_{s:e} - \Theta_0\right\rangle \nonumber\\
        &+ \lambda_{s,e} (e-s) \left\langle \widehat{V}, \, \Theta_0 - \widetilde{\Theta}_{s:e}\right\rangle \nonumber \\
		= & \sum_{t = s+1}^e \left\langle \left(Y(t) - \Theta_0\right)_{\Omega(t)} + \left(\widehat{\Theta} - \widetilde{\Theta}_{s:e}\right)_{\overline{\Omega(t)}}, \, \widetilde{\Theta}_{s:e} - \Theta_0 \right\rangle - \lambda_{s,e} (e-s) \left\langle \widehat{V}, \, \widetilde{\Theta}_{s:e} - \Theta_0\right\rangle \nonumber \\
		\geq & \sum_{t = s+1}^e \left\langle \left(Y(t) - \Theta_0\right)_{\Omega(t)} + \left(\widehat{\Theta} - \widetilde{\Theta}_{s:e}\right)_{\overline{\Omega(t)}}, \, \widetilde{\Theta}_{s:e} - \Theta_0 \right\rangle \nonumber \\
        & \hspace{1cm} - \sum_{t = s+1}^e \left\langle \left\{Y_{\Omega}(t) + \widehat{\Theta}_{\overline{\Omega(t)}}\right\} - \widetilde{\Theta}_{s:e}, \, \widetilde{\Theta}_{s:e} - \Theta_0\right\rangle \nonumber \\
		= & \sum_{t = s+1}^e \left\langle\left (\widetilde{\Theta}_{s:e} - \Theta_0\right)_{\Omega(t)}, \, \widetilde{\Theta}_{s:e} - \Theta_0\right\rangle = \sum_{t = s+1}^e \left\|\left(\widetilde{\Theta}_{s:e} - \Theta_0\right)_{\Omega(t)}\right\|_{\mathrm{F}}^2, \label{eq-proof-prop3-step1-2}
	\end{align}
    where the inequality follows from \eqref{eq-proof-prop3-step1-1} and the second identity follows from the fact that $Y_{\Omega}(t)$'s are only observed on $\Omega(t)$'s.
\\
Equation~\eqref{eq-proof-prop3-step1-2} implies that
	\begin{align}
		& \frac{1}{e-s} \sum_{t = s+1}^e \left\|\left(\widetilde{\Theta}_{s:e} - \Theta_0\right)_{\Omega(t)}\right\|_{\mathrm{F}}^2 \nonumber \\ 
		\leq & \left|\frac{1}{e-s} \sum_{t = s+1}^e \left\langle (Y(t) - \Theta_0)_{\Omega(t)}, \, \widetilde{\Theta}_{s:e} - \Theta_0\right\rangle \right|\label{eq-decomp-1} \\
		& + \left|\frac{1}{e-s} \sum_{t = s+1}^e \left\langle \left(\widehat{\Theta} - \widetilde{\Theta}_{s:e}\right)_{\overline{\Omega(t)}}, \, \widetilde{\Theta}_{s:e} - \Theta_0\right\rangle \right|\label{eq-decomp-2} \\
		& + \lambda_{s,e} \left\langle \widehat{V}, \, \Theta_0 - \widetilde{\Theta}_{s:e}\right\rangle. \label{eq-decomp-3} 
	\end{align}
\\
For term \eqref{eq-decomp-1}, recalling the definition of $\Sigma(t)$ in \eqref{eq-Sigma-mat-def}, we have that
	\begin{align}
		\eqref{eq-decomp-1} =& \left|\Bigg\langle \frac{1}{e-s} \sum_{t = s+1}^e \Sigma(t), \, \widetilde{\Theta}_{s:e} - \Theta_0 \Bigg\rangle\right|	\leq \left\|\frac{1}{e-s} \sum_{t = s+1}^e \Sigma(t)\right\|_{\mathrm{op}} \left\|\widetilde{\Theta}_{s:e} - \Theta_0\right\|_* \nonumber \\
        \leq & \frac{\lambda_{s,e}}{3} \left\|\widetilde{\Theta}_{s:e} - \Theta_0\right\|_*, \label{eq-decomp-1-2}
	\end{align}
	due to the condition on $\lambda_{s,e}$ in \eqref{eq-prop-3-lambda-cond}.
\\
For term \eqref{eq-decomp-2}, we have that 
	\begin{align} \label{eq-decomp-2-2}
		\eqref{eq-decomp-2}	\leq \left\|\frac{1}{e-s} \sum_{t = s+1}^e \left(\widehat{\Theta} - \widetilde{\Theta}_{s:e}\right)_{\overline{\Omega(t)}}\right\|_{\mathrm{op}} \left\|\widetilde{\Theta}_{s:e} - \Theta_0\right\|_* \leq \frac{\lambda_{s,e}}{3} \left\|\widetilde{\Theta}_{s:e} - \Theta_0\right\|_*,
	\end{align}
	due to the stopping rule of \Cref{alg-main-klopp}.
\\
For term \eqref{eq-decomp-3}, due to the monotonicity of sub-differentials of convex functions, we have that for any $V \in \partial \|\Theta_0\|_*$, 
    \[
        \left\langle \widehat{V}, \Theta_0 - \widetilde{\Theta}_{s:e} \right\rangle \leq \left\langle V, \Theta_0 - \widetilde{\Theta}_{s:e} \right\rangle.
    \]
    We then have that
	\begin{align} \label{eq-decomp-3-2}
		\eqref{eq-decomp-3} \leq \lambda_{s,e} \left(\left\|\mathbf{P}_{\Theta_0} (\Theta_0 - \widetilde{\Theta}_{s:e})\right\|_* - \left\|\mathbf{P}_{\Theta_0} (\widetilde{\Theta}_{s:e})^{\perp}\right\|_*\right).
	\end{align}	
\\
Combining \eqref{eq-decomp-1-2}, \eqref{eq-decomp-2-2} and \eqref{eq-decomp-3-2}, recalling the notation introduced in \eqref{eq-mathbfP-def}, we have that
	\begin{align}\label{eq-frobenius-nuclear}
		\frac{1}{e-s} \sum_{t = s+1}^e \left\|\left(\widetilde{\Theta}_{s:e} - \Theta_0\right)_{\Omega(t)}\right\|_{\mathrm{F}}^2 \leq \frac{5\lambda_{s,e}}{3} \left\|\mathbf{P}_{\Theta_0}(\Theta_0 - \widetilde{\Theta}_{s:e})\right\|_* \leq \frac{5\lambda_{s,e}}{3} \sqrt{2\mathrm{rank}(\Theta_0)} \left\|\Theta_0 - \widetilde{\Theta}_{s:e}\right\|_{\mathrm{F}}
	\end{align}
	and
	\[
		\frac{\lambda_{s,e}}{3}\left\|\mathbf{P}_{\Theta_0}(\widetilde{\Theta}_{s:e})^{\perp}\right\|_* \leq \frac{5\lambda_{s,e}}{3} \left\|\mathbf{P}_{\Theta_0}(\Theta_0 - \widetilde{\Theta}_{s:e})\right\|_*.
	\]
	We therefore have that
	\[
		\left\|\widetilde{\Theta}_{s:e} - \Theta_0\right\|_* \leq 6\left\|\mathbf{P}_{\Theta_0}(\Theta_0 - \widetilde{\Theta}_{s:e})\right\|_* \leq \sqrt{72 \mathrm{rank}(\Theta_0)} \left\|\widetilde{\Theta}_{s:e} - \Theta_0\right\|_{\mathrm{F}}.
	\]

\medskip  

\noindent \textbf{Step 2.}  For any matrix $A \in \mathbb{R}^{n \times n}$, recall the notation given in \eqref{eq-weightedFrob-def}
	\[
		\|A\|^2_{L_2(\Pi), s:e} = \frac{1}{e-s} \sum_{t = s+1}^e \sum_{(i, j)} \Pi_{ij}(t) A_{ij}^2 \geq q_1 \|A\|^2_{\mathrm{F}}.
	\]	
	For absolute constants $C, c_1 > 0$, let
	\[
		R_{s:e} = \frac{C\log(3/\delta)}{c_1q_1 (e-s) \log(12)}.
	\]  
  We now consider two cases.
	
\medskip	
\noindent \textbf{Step 2.1.} If $\|\widetilde{\Theta}_{s:e} - \Theta_0\|^2_{L_2(\Pi), s:e} < R_{s:e}$, then we complete the result.

\medskip	
\noindent \textbf{Step 2.2.} If $\|\widetilde{\Theta}_{s:e} - \Theta_0\|^2_{L_2(\Pi), s:e} \geq R_{s:e}$, then we consider the following set
	\[
		\mathcal{C}(r, e-s) = \Big\{A \in \mathbb{R}^{n \times n}: \, \|A\|_{\infty} = 1, \, \|A\|_{L_2(\Pi), s:e}^2 \geq \frac{C\log(3/\delta)}{c_1q_1 (e-s) \log(12)}, \, \|A\|_* \leq \sqrt{r}\|A\|_{\mathrm{F}}\Big\}.
	\]
 Due to the conditions in \Cref{alg-main-klopp}, $\|\widetilde{\Theta}_{s:e} - \widehat{\Theta}\|_{\infty} \leq \vartheta$ and $\|\widehat{\Theta}\|_{\infty} \leq \vartheta$, which implies that $\|\widetilde{\Theta}_{s:e} - \Theta_0\|_{\infty} \leq \|\widetilde{\Theta}_{s:e} - \widehat{\Theta}\|_{\infty} + \|\widehat{\Theta}\|_{\infty} + \|\Theta_0\|_{\infty} \leq 3\vartheta$.  
	We then have that
	\[
		(3\vartheta)^{-1}(\widetilde{\Theta}_{s:e} - \Theta_0) \in \mathcal{C}(72\mathrm{rank}(\Theta_0), e-s).
	\]
	Due to \Cref{lem-9-klopp_dep}, with probability at least $1 - \delta$, one has that
	\begin{align*}
		& \|\widetilde{\Theta}_{s:e} - \Theta_0\|_{L_2(\Pi), s:e}^2\\
	        \leq & \frac{2}{e-s} \sum_{t = s+1}^e \|(\widetilde{\Theta}_{s:e} - \Theta_0)_{\Omega(t)}\|_{\mathrm{F}}^2 + C\frac{\vartheta ^2}{q_1}\mathrm{rank}(\Theta_0)\left\{\mathbb{E}\left[\|\Sigma_{R,s:e}\|_{\mathrm{op}}\right]\right\}^2\\
		\leq & \frac{10\lambda_{s,e}}{3} \sqrt{2\mathrm{rank}(\Theta_0)} \|\Theta_0 - \widetilde{\Theta}_{s:e}\|_{\mathrm{F}} + C\frac{\vartheta^2}{q_1}\mathrm{rank}(\Theta_0)\left\{\mathbb{E}\left[\|\Sigma_{R,s:e}\|_{\mathrm{op}}\right]\right\}^2\\
		\leq & \frac{C_1\lambda_{s,e}^2}{q_1} \mathrm{rank}(\Theta_0) + C_2q_1 \|\Theta_0 - \widetilde{\Theta}_{s:e}\|_{\mathrm{F}}^2 + C\frac{\vartheta^2}{q_1}\mathrm{rank}(\Theta_0)\left\{\mathbb{E}\left[\|\Sigma_{R,s:e}\|_{\mathrm{op}}\right]\right\}^2\\
  	\leq & \frac{C_1\lambda_{s,e}^2}{q_1} \mathrm{rank}(\Theta_0) + C_2\|\widetilde{\Theta}_{s:e} - \Theta_0\|_{L_2(\Pi), s:e}^2 + C\frac{\vartheta^2}{q_1}\mathrm{rank}(\Theta_0)\left\{\mathbb{E}\left[\|\Sigma_{R,s:e}\|_{\mathrm{op}}\right]\right\}^2,
	\end{align*}
where $C_1 > 0$ and $C_2 \in (0,1)$ are some absolute constant, the second inequality follows from \eqref{eq-frobenius-nuclear}, and the last inequality follows from \eqref{eq-weightedFrob-def}. 
\\
	Therefore we have that
	\[
		\|\widetilde{\Theta}_{s:e} - \Theta_0\|_{L_2(\Pi)}^2 \leq \frac{C}{q_1} \mathrm{rank}(\Theta_0)\left[\lambda_{s,e}^2 + \vartheta^2 \left\{\mathbb{E}\left[\|\Sigma_{R,s:e}\|_{\mathrm{op}}\right]\right\}^2\right].
	\]
	The final result holds due to the fact that $Q_{ij}(t) \geq q_1$, for any $(i, j) \in \{1, \ldots, n\}^{\otimes 2}$ and $t \in \mathbb{N}^*$.

\end{proof}

\subsubsection[]{Proof of \Cref{cor-main-klopp-cor-3}}
\begin{proof}[Proof of \Cref{cor-main-klopp-cor-3}]
This is an immediate consequence of \Cref{thm-main-klopp-cor-3} and \Cref{lem-lem4-klopp_dep}.
\end{proof}

\subsubsection{Change point analysis}

\begin{lemma}\label{lem-large-prob-event-2}
Under all the conditions and notation in \Cref{cor-main-klopp-cor-3}, define
    \begin{equation}\label{eq-event-G-def}
        \mathcal{E} = \left\{\forall s, t \in \mathbb{N}, \, t > s, \, t \geq 2: \, \|\widehat{\Theta}_{s:t} - \Theta_0\|_{\mathrm{F}} < \zeta_{s, t}\right\},
    \end{equation}
    where 
    \begin{align*}  
    \zeta_{s, t} = \sqrt{\frac{C_{\mathrm{noise}}rn\vartheta\log(n)\log(1/{\delta}_t)}{q_1^2(t-s)}}\sqrt{\max\left\{1,nq_2^2\right\}},
    \end{align*}
    with $\alpha \in (0, 1)$, 
    \begin{equation}\label{eq-lemma-4-delta-t-def}
        \delta_t = \frac{\alpha}{16 t^2}
    \end{equation}
    and $C_{\mathrm{noise}} > 0$ being a sufficiently large absolute constant.  It holds that
    \[
        \mathbb{P}\{\mathcal{E}\} \geq 1 - \alpha.
    \]
\end{lemma}

\begin{proof}[Proof of \Cref{lem-large-prob-event-2}]
Let $\mathcal{E}^c$ be the complement of $\mathcal{E}$.  We thus have that
    \begin{align}
        \mathbb{P}\{\mathcal{E}^c\} & = \mathbb{P}\left\{\exists s, t \in \mathbb{N}, \, t > s, \, t \geq 2 : \, \|\widehat{\Theta}_{s:t} - \Theta_0\|_{\mathrm{F}} \geq \zeta_{s, t}\right\}  \nonumber \\
        & \leq \sum_{j = 1}^{\infty} \mathbb{P}\left\{\max_{2^j \leq t < 2^{j+1}} \max_{0 \leq s < t} \left\{\|\widehat{\Theta}_{s:t} - \Theta_0\|_{\mathrm{F}} - \zeta_{s, t}\right\} \geq 0\right\} \nonumber \\
        & \leq \sum_{j = 1}^{\infty} 2^j \max_{2^j \leq t < 2^{j+1}} \mathbb{P}\left\{ \max_{0 \leq s < t} \left\{\|\widehat{\Theta}_{s:t} - \Theta_0\|_{\mathrm{F}} - \zeta_{s, t} \right\} \geq 0 \right\} \nonumber \\
        & \leq \sum_{j = 1}^{\infty} 2^{2j+1} \max_{2^j \leq t < 2^{j+1}} \max_{0 \leq s < t} \mathbb{P}\left\{\|\widehat{\Theta}_{s:t} - \Theta_0\|_{\mathrm{F}} \geq \zeta_{s, t}\right\}. \label{eq-lemma-s4-pf-1}
    \end{align}
Due to \Cref{cor-main-klopp-cor-3}, for any integer pair $0 \leq s < t$, such that the interval length $(e-s)$ satisfying \begin{align*}
    \frac{\sqrt{e-s}}{(\log(e-s))^2} \geq C\sqrt{n\log(n)},
\end{align*}
we have that for any $\delta > 0$,
\begin{align*}
\mathbb{P}\Bigg(\|\widehat{\Theta}_{s:t} - \Theta_0\|_{\mathrm{F}} \geq \sqrt{\frac{C_{\mathrm{noise}}rn\vartheta \log(n)\log(1/\delta)}{q_1^2(t-s)}}\sqrt{\max\left\{1,nq_2^2\right\}}\Bigg) < 2\delta
\end{align*}
    Choosing 
    \begin{align*}  
    \widetilde{\zeta}_{s, t} = \sqrt{\frac{C_{\mathrm{noise}}rn\vartheta\log(n)\log(1/\widetilde{\delta}_t)}{q_1^2(t-s)}}\sqrt{\max\left\{1,nq_2^2\right\}},
    \end{align*}
    and
    \[\widetilde{\delta}_t = \frac{\alpha \log^2(2)}{4 \{\log(t) + \log(2)\}^2 t^2},
    \]
    due to \eqref{eq-lemma-s4-pf-1}, we have that
    \begin{align*}
        \mathbb{P}\{\mathcal{E}^c\} & \leq \sum_{j = 1}^{\infty} 2^{2j+1} \max_{2^j \leq t < 2^{j+1}} \max_{0 \leq s < t} \mathbb{P}\left\{\|\widehat{\Theta}_{s:t} - \Theta_0\|_{\mathrm{F}} \geq \widetilde{\zeta}_{s, t}\right\}\\
        & \leq \sum_{j = 1}^{\infty} 2^{2j+1} \max_{2^j \leq t < 2^{j+1}} \max_{0 \leq s < t}  \frac{\alpha \log^2(2)}{2 \{\log(t) + \log(2)\}^2 t^2} \\
        & \leq \alpha \sum_{j = 1}^{\infty} 2^{2j+1} \frac{1}{2\left\{\frac{\log(2^j)}{\log(2)} + 1\right\}^2 2^{2j}} \leq \alpha \sum_{j = 1}^{\infty} \frac{1}{(j+1)^2} \leq \alpha \sum_{j = 1}^{\infty} \frac{1}{j(j+1)} = \alpha,
    \end{align*}
    where the first inequality holds since $\delta_t \leq \widetilde{\delta}_t$, with $t \geq 2$, $\delta_t$ defined in \eqref{eq-lemma-4-delta-t-def}.
\end{proof}

\subsubsection[]{Proof of \Cref{thm-upper-bound}}

\begin{proof}[Proof of \Cref{thm-upper-bound}]
The proof is conducted in the event $\mathcal{E}$, with the event $\mathcal{E}$ defined as
\[
    \mathcal{E} = \left\{\forall s, t \in \mathbb{N}, \, t > s, \, t \geq 2: \, \|\widehat{\Theta}_{s:t} - \Theta_0\|_{\mathrm{F}} < \zeta_{s, t}\right\}.
\]
It follows from \Cref{lem-large-prob-event-2} that $\mathbb{P}\{\mathcal{E}\} > 1 - \alpha$.  

 	\begin{align*}
	\|\widehat{\Theta}_{s:e} - \Theta_0\|_{\mathrm{F}} \leq \sqrt{\frac{C_{\mathrm{noise}}rn\vartheta\log(n)\log(1/\delta)}{q_1^2(e-s)}}\sqrt{\max\left\{1,nq_2^2\right\}},
	\end{align*}

For any $t \leq \Delta$, we have that 
    \begin{align*}
        & \|\widehat{\Theta}_{0:s} - \widehat{\Theta}_{s:t}\|_{\mathrm{F}} \leq \|\widehat{\Theta}_{0:s} - \Theta_{\Delta}\|_{\mathrm{F}} + \|\widehat{\Theta}_{s:t} - \Theta_{\Delta}\|_{\mathrm{F}} \\
        \leq & \left(\sqrt{\frac{C_{\mathrm{noise}}rn\vartheta\log(n)\log(32s^2/\alpha)}{q_1^2s}} + \sqrt{\frac{C_{\mathrm{noise}}rn\vartheta\log(n)\log(32t^2/\alpha)}{q_1^2(t-s)}}\right) \sqrt{\max\{1, nq_2^2\}}\\
        \leq & C_{\varepsilon} \frac{\sqrt{ r\vartheta n\log(n)\max\{1, nq_2^2\}} }{q_1}\left(\sqrt{\frac{\log(s/\alpha)}{s}} + \sqrt{\frac{\log(t/\alpha)}{t-s}}\right)\\
        =& \varepsilon_{s, t},
    \end{align*}
    where $C_{\varepsilon} > 0$ is a large enough absolute constant depending only on $C_{\mathrm{noise}}$.  We thus have that $\widehat{\Delta} > t$ and consequently $\widehat{\Delta} > \Delta$.

Now we consider $t > \Delta$.  Let 
    \[
        \widetilde{\Delta} = \min\{t > \Delta: \, \|\widehat{\Theta}_{0: \Delta} - \widehat{\Theta}_{\Delta: t}\|_{\mathrm{F}}^2 \geq \varepsilon_{\Delta, t}\}.
    \]
    From the definition of $\widehat{\Delta}$, it holds that $\widetilde{\Delta} \geq \widehat{\Delta}$.  It now suffices to upper bound $(\widetilde{\Delta} - \Delta)_+$.
    
On the event $\mathcal{E}$, we have that
    \begin{align}
        & \|\widehat{\Theta}_{0:\Delta} - \widehat{\Theta}_{\Delta: \widetilde{\Delta}}\|_{\mathrm{F}} = \|\widehat{\Theta}_{0:\Delta} - \Theta_{\Delta} + \Theta_{\Delta} - \Theta_{\Delta + 1} - \Theta_{\Delta+1} - \widehat{\Theta}_{\Delta: \widetilde{\Delta}}\|_{\mathrm{F}} \nonumber \\
        \geq & \|\Theta_{\Delta} - \Theta_{\Delta+1}\|_{\mathrm{F}}  - \left(\|\widehat{\Theta}_{0:\Delta} - \Theta_{\Delta}\|_{\mathrm{F}} + \|\widehat{\Theta}_{\Delta: \widetilde{\Delta}} - \Theta_{\Delta+1}\|_{\mathrm{F}} \right) \nonumber \\
        \geq & \kappa -  C_{\varepsilon} \frac{\sqrt{ r\vartheta n\log(n)\max\{1, nq_2^2\}} }{q_1}\Bigg(\sqrt{\frac{\log(\Delta/\alpha)}{\Delta}}  + \sqrt{\frac{\log\{(d+\Delta)/\alpha\}}{d}}\Bigg), \label{eq-proof-main-11}
    \end{align}
    where we let $d = \widetilde{\Delta} - \Delta$ in the last inequality.
    It now suffices to show that with the choice of 
\begin{align*}
    d = C_d  \frac{r\vartheta n\log(n)\log(\Delta/\alpha)}{q_1^{2}\kappa^2}\max\{1, nq_2^2\},
\end{align*}
    where $C_d > 0$ is a large enough absolute constant depending only on $C_{\varepsilon}$, it holds that $\eqref{eq-proof-main-11} \geq \varepsilon_{\Delta, \Delta + d}$, which is equivalent to 
    \[
        \kappa \geq 2\varepsilon_{\Delta, \Delta + d}.
    \]
    
Due to Assumption~\ref{assume-snr}, i.e.~\begin{align*}
    \kappa^2 \Delta \geq C_{\mathrm{SNR}}\frac{r\vartheta n\log(n)\log(\Delta/\alpha)}{q_1^{2}}\max\{1, nq_2^2\},
\end{align*}
    where $C_{\mathrm{SNR}} \geq C_d$ is a large enough absolute constant, we have that $\Delta > d$ and $\Delta > n$, which leads to 
\begin{align*}
    &2\varepsilon_{\Delta, \Delta+d}\\
    =& 2C_{\varepsilon}\frac{\sqrt{r\vartheta n\log(n)}}{q_1} \left(\sqrt{\frac{\log(\Delta/\alpha)}{\Delta}} + \sqrt{\frac{\log(\Delta+d/\alpha)}{d}}\right)\sqrt{\max\{1, nq_2^2\}}\\
    \leq&  2C_{\varepsilon}\frac{\sqrt{r\vartheta n\log(n)}}{q_1} \left(\sqrt{\frac{\log(\Delta/\alpha)}{\Delta}} + \sqrt{\frac{\log(2\Delta/\alpha)q_1^{2}\kappa^2}{C_d r\vartheta n\log(n)\log(\Delta/\alpha)\max\{1, nq_2^2\}}}\right)\sqrt{\max\{1, nq_2^2\}}\\ 
    \leq& \kappa,
\end{align*}
where the last inequality holds for a large enough $C_{SNR}$ depending on $C_{\epsilon}$ and $C_d$.
\end{proof}

\subsection[]{Proof of \Cref{lemma:ergodic}}
\begin{proof}[Proof of \Cref{lemma:ergodic}]
    Note that for any $i \in \{1, \ldots, n\}$, $\{X_i(t)\}_{t \in \mathbb{N}^*}$ is a Markov chain with state space $\mathcal{X}$. By Theorem 3.7 in \cite{bradley2005basic}, it suffices to verify Doeblin’s condition \citep[e.g.][]{rosenthal1995convergence}, which holds on the transition probability $\mathbb P(\cdot, \cdot)$, if there exists an $\epsilon > 0$ and a probability measure $\mu(\cdot)$, such that for all $x \in \mathcal{X}$ and measurable subsets $A \subseteq \mathcal{X}$,
    $\mathbb P(x, A) \geq \epsilon\mu(A)$.

    We write \eqref{eq:markov_latentpostion} equivalently as
    \[
        X_i(t) = (1-U)\cdot X_i(t-1) + U\cdot Z,
    \]
    where $U \sim \mathrm{Bernoulli}(1-\rho)$, $Z \sim F$, and $U$ is independent of $X_i(t-1)$ and $Z$. We have for any $A \subseteq \mathcal{X}$
    \begin{align*}
        \mathbb P(x, A) =& \mathbb P(X_i(t) \in A| X_i(t-1) = x)\\
        =& \mathbb P(X_i(t) \in A, U = 1| X_i(t-1) = x)\\
        &+ \mathbb P(X_i(t) \in A, U = 0| X_i(t-1) = x)\\
        \geq& \mathbb P(X_i(t) \in A, U = 1| X_i(t-1) = x)\\
        =& \mathbb P(Z \in A, U = 1)\\
        =& (1-\rho)\mathbb P(Z \in A).
    \end{align*}
    Thus, we have $\epsilon = 1-\rho > 0$, which concludes the proof.
\end{proof}

\section{Auxiliary lemmas}
This section collects some existing results which are used in our proofs. We use $\phi$-mixing throughout our paper. Howeve, there are other variants of mixing coefficients, such as the $\alpha-$ and $\beta$-mixing coefficients \cite[e.g.][]{bradley2005basic,doukhan2012mixing}. Note that the conditions based on the $\alpha$-mixing are weaker than that of $\beta$-mixing, which are further weaker than that of $\phi$-mixing. More precisely, it follows that $2\alpha(\ell)  \leq \beta(\ell) \leq \phi(\ell)$ for any $\ell \in \mathbb{N}$ \cite[e.g.][]{bradley2005basic}.
\iffalse
\begin{lemma}[Theorem 3.4 in \cite{bradley2005basic}]\label{thm:phi-mixing}
    Suppose $\{X(t)\}_{t \in \mathbb{N}^*}$ is a strictly stationary Markov chain which is irreducible and aperiodic. If $\phi(\ell) < 1$ for some $\ell \geq 1$, then $\phi(\ell) \to 0$ at least exponentially fast as $\ell \to \infty$.
\end{lemma}
\Cref{thm:phi-mixing} considers Markov chains whose state space not necessary be finite or countable. Note that the $\phi$-mixing coefficients satisfy that $0 \leq \phi(\ell) \leq 1$ and $2\alpha(\ell)  \leq \beta(\ell) \leq \phi(\ell)$ for any $\ell \in \mathbb{N}$. \Cref{thm:phi-mixing} states that for a strictly stationary irreducible and aperiodic Markov chain, excluding $\phi(\ell) = 1$ for some $\ell \geq 1$ guarantees $\phi(\ell)$ decays to zero exponentially as $\ell \to \infty$.
\fi

\begin{lemma}[Theorem 5.2 in \citealt{bradley2005basic}]\label{lemma:mixing_indep_comp}
Denote $X^{(k)} = \{X^{(k)}(t)\}_{t \in \mathbb{N}^*}$, for $k \in \mathbb N^*$. Suppose these sequences, $X^{(k)}$, are mutually independent across $k$. Suppose that for each $t \in \mathbb{N}^*$, $h_t: \mathbb{R}\times \mathbb{R} \times \mathbb{R} \times \dots \mapsto \mathbb{R}$ is a Borel function. Define the sequence $X = \{X(t)\}_{t \in \mathbb{N}^*}$ of random variables with $X(t) = h_t\big(X^{(1)}(t),X^{(2)}(t),X^{(3)}(t),\dots\big)$, for $t \in \mathbb{N}^*$. Then for each $\ell \geq 1$, it follows that $\phi_X(\ell) \leq 1-\prod_{k=1}^{\infty}\big(1-\phi_{X^{(k)}}(\ell)\big) \leq \sum_{k = 1}^{\infty}\phi_{X^{(k)}}(\ell)$.
\end{lemma}

\iffalse
\begin{lemma}[Theorem 1 in \citealt{merlevede2011bernstein}]\label{lemma:bernstein}
    Let $\{X_i\}_{i \geq 1}$ be a sequence of centered $\mathbb{R}$-valued random variables with $\alpha$-mixing coefficients $\alpha(\ell)$. Suppose there exist constants $c, \gamma_1, \gamma_2 > 0$ and $1/\gamma = 1/\gamma_1 + 1/\gamma_2 > 1$, such that for any $\ell \geq 1$
    \begin{equation*}
        \alpha(\ell) \leq \exp(-c\ell^{\gamma_1}),
    \end{equation*}
    and for any $u > 0$ 
    \begin{equation}\label{eq:marginal_tail}
        \sup_{i \geq 1}\mathbb{P}(|X_i| > u) \leq \exp(1-u^{\gamma_2}).
    \end{equation}
    Then for any $n \geq 4$, there exist positive absolute constants $C_1$ and $C_2$ depending only on $c$, $\gamma_1$ and $\gamma_2$ such that for any $x \geq 1$
    \begin{equation*}
        \mathbb{P}\bigg(n^{-1/2}\bigg|\sum_{i = 1}^nX_i\bigg| \geq x\bigg) \leq n\exp\bigg(-\frac{x^{\gamma}n^{\gamma/2}}{C_1}\bigg) + \exp\bigg(-\frac{x^2n}{C_2(1+nV^2)} \bigg),
    \end{equation*}
    where
        \begin{equation*}
        V^2 = \sup_{i \geq 1}\bigg\{\var(X_i) + 2\sum_{j > i}|\cov(X_i, X_j)|\bigg\}.
    \end{equation*}
\end{lemma}
\fi
%%%%%%%%%%%%%%%%%%%%%%%%%%%%%%%%%%%%%%%%%%%%%%%%%%%%%%%%%%%%%%%%%%%%%%%%%%%%
\begin{lemma}[Theorem 1 in \cite{banna2016bernstein}]\label{lemma:matrix_bernstein}
Let $\{M_i\}_{i \in \mathbb{N}^*}$ is a family of self-adjoint random matrices of size $d$. Assume that there exists a constant $c > 0$ such that for any $\ell \geq 1$, $\beta_M(\ell) \leq \exp(1-c\ell)$, and there exist a positive constant $D$ such that for any $i \in \mathbb{N}^*$,
\begin{align*}
    \mathbb{E}[M_i] = 0 \quad \text{and} \quad \|M_i\|_{\mathrm{op}} \leq D \quad \text{almost surely}.
\end{align*}
Then there exists an absolute constant $C$ such that for any $x > 0$ and any integer $n \geq 2$,
\[
\mathbb{P}\left(\left\|\sum_{i=1}^nM_i\right\|_{\mathrm{op}} \geq x\right) \leq d\exp\left(-\frac{Cx^2}{\nu^2n + c^{-1}D^2 + xD\gamma(c,n)}\right),
\]
where
\[
\nu^2 = \sup_{\mathcal{K} \subseteq\{1, \dots, n\}}\frac{1}{|\mathcal{K}|}\left\|\mathbb{E}\left(\sum_{i \in \mathcal{K}}M_i\right)^2\right\|_{\mathrm{op}}
\]
and
\[
\gamma(c,n) = \frac{\log n}{\log 2}\max\left\{2, \frac{32\log n}{c\log 2}\right\}.
\]
\end{lemma}

\begin{lemma}[Proposition 2.6.1 in \citealt{vershynin2018high}]\label{lemma:subGaussian_ind_sum}
    Let $X_1, \dots, X_n$ be independent, mean zero, sub-Gaussian random variables. Then $\sum_{i=1}^n X_i$ is also a sub-Gaussian random variable, and
    \[
        \Big\|\sum_{i=1}^n X_i \Big\|_{\psi_2}^2 \leq C\sum_{i = 1}^n\|X_i\|_{\psi_2}^2,
    \]
    where $C > 0$ is a universal constant.
\end{lemma}

\begin{lemma}[Theorem 4.4.5 \& Corollary 4.4.8 in \citealt{vershynin2018high}]\label{thm:op_norm_tail}
    (a) Let $X \in \mathbb{R}^{m \times n}$ be a random matrix whose entries $X_{ij}$ are independent, centred and sub-Gaussian. Then, for any $t > 0$, we have that
    \[
        \|X\|_{\mathrm{op}} \leq CK(\sqrt{m} + \sqrt{n} + t),
    \]
    with probability at least $1 - \exp(-t^2)$, and $K = \max_{i,j}\|X_{ij}\|_{\psi_2}$.

    (b) Let $X \in \mathbb{R}^{n \times n}$ be a symmetric random matrix whose upper diagonal entries $X_{ij}$ are independent, centred and sub-Gaussian. Then, for any $t > 0$, we have that
    \[
        \|X\|_{\mathrm{op}} \leq CK(\sqrt{n} + t),
    \]
    with probability at least $1 - \exp(-t^2)$, and $K = \max_{i,j}\|X_{ij}\|_{\psi_2}$.
\end{lemma}

The next theorem gives a variant of Talagrand's concentration inequality under $\phi$-mixing condition.
\begin{lemma}[Corollary 4 in \citealt{samson2000concentration}]\label{thm-talagrand_dep}
Suppose that $f: \, [-1, 1]^n \to \mathbb{R}$ is convex and $L$-Lipschitz continus with respect to the $l_2$-norm.  Let $\epsilon_1, \ldots, \epsilon_n$ be a sequence of random variables taking values in $[-1, 1]$.  Let $Z = f(\epsilon_1, \ldots, \epsilon_n)$.  Then for any $t \geq 0$, it holds that 
	\[
		\mathbb{P}\left\{|Z - \mathbb{E}(Z)| \geq t\right\} \leq 2\exp\left\{-\frac{t^2}{2L^2\|\Gamma\|^2_{\mathrm{op}}}\right\},
	\]
where $\Gamma = (\gamma_{ij})_{1\leq i,j \leq n}$ is a lower triangular matrix having $1$s along its diagonal. For $1 \leq i < j \leq n$, let $\epsilon_i^j$ represent the vector $(\epsilon_i, \dots, \epsilon_j)$ and $\mathcal{L}(\epsilon_{j}^n|\epsilon_1^i = y_1^i)$ denote the law of $\epsilon_j^n$ conditionally to $\epsilon_1^i = y_1^i$. Then set
\begin{align*}
    (\gamma_{ij})^2 = 2\esssup_{y_1^i \in \mathbb R^i, \mathcal{L}(\epsilon_1^i)} \left\|\mathcal{L}\left(\epsilon_{j}^n|\epsilon_1^i = y_1^i\right) - \mathcal{L}\left(\epsilon_j^n\right)\right\|_{\mathrm{TV}},
\end{align*}
where $\esssup_{y_1^i \in \mathbb R^i, \mathcal{L}(\epsilon_1^i)}$ is the essential supremum with respect to the measure $\mathcal{L}(\epsilon_1^i)$ and $\|\cdot\|_{\mathrm{TV}}$ is the total variation of a signed measure. 

Further assume that the $\phi$-mixing coefficients of $\epsilon_1, \ldots, \epsilon_n$ decay exponentially, i.e.~there exists some $C > 0$ and $\rho \in (0,1)$ such that $\phi(\ell) \leq C\rho^\ell$, for any $\ell \in \mathbb{N}$, then it follows that $\|\Gamma\|_{\mathrm{op}} \leq \sum_{\ell = 1}^{\infty}\sqrt{\phi(\ell)} < \infty$.
\end{lemma}

\begin{remark}
Note that the above lemma is stated under the $\phi$-mixing conditions, which are stronger than the $\alpha$ or $\beta$-mixing conditions with the same decay rate. This lemma is the main reason for which we impose the $\phi$-mixing conditions. To the best of our knowledge, it is still unknown if the Talagrand's concentration inequality could be generalised under the more general $\alpha-$ or $\beta$-mixing conditions.
\end{remark}

For $r \in \mathbb{N}^*$ and integer pair $1 \leq s < e$, define the set 
    \begin{align}\label{eq-def-c-r-es}
		\mathcal{C}(r, e-s) = \Big\{A \in \mathbb{R}^{n \times n}: \, \|A\|_{\infty} = 1, \, \|A\|_{L_2(\Pi), s:e}^2 \geq \frac{C\log(3/\delta)}{c_1q_1 (e-s) \log(12)}, \, \|A\|_* \leq \sqrt{r}\|A\|_{\mathrm{F}}\Big\},
	\end{align}
    where for any matrix $A \in \mathbb{R}^{n \times n}$, the norm $\|A\|_{L_2(\Pi), s:e} = (e-s)^{-1}\sum_{t = s+1}^e\|A\|_{L_2(\Pi(t))}$ is defined in \eqref{eq-weightedFrob-def},
    and $\{\Pi(t)\}_{t \in \mathbb N^*}$ is a sequence of missingness probability matrices defined in \Cref{def-sample}. Recall that 
    \[
        \Sigma_{R,s:e} = \frac{1}{e-s}\sum_{t = s+1}^e\Sigma_R(t),
    \]
    where $\Sigma_R(t)$ is defined in \eqref{eq-Sigma-prime-maat-def}.

\begin{lemma}\label{lem-9-klopp_dep}
For any integer pair $1 \leq s < e$ and for any $A \in \mathcal{C}(r, e-s)$ as defined in \eqref{eq-def-c-r-es}, it holds that
    \begin{align*}
	&\mathbb{P}\left(\frac{1}{2}\|A\|^2_{L_2(\Pi), s:e} - \frac{54r}{q_1} \left\{\mathbb{E}\left[\|\Sigma_{R, s:e}\|_{\mathrm{op}}\right]\right\}^2 \leq \frac{1}{e-s} \sum_{t = s+1}^e \|A_{\Omega(t)}\|_{\mathrm{F}}^2\right)
    \geq 1-\delta.
	\end{align*}
\end{lemma}
\begin{proof}[Proof of \Cref{lem-9-klopp_dep}]
Let
	\[
		R = \frac{54r}{q_1}\left\{\mathbb{E}\left[\|\Sigma_{R, s:e}\|_{\mathrm{op}}\right]\right\}^2,
	\]
	and the event
	\[
		\mathcal{E} = \left\{\exists A \in \mathcal{C}(r, e-s) \,\mbox{ s.t. } \left|\frac{1}{e-s}\sum_{t = s+1}^e \|A_{\Omega(t)}\|_{\mathrm{F}}^2 - \|A\|^2_{L_2(\Pi), s:e}\right| > \frac{1}{2} \|A\|^2_{L_2(\Pi), s:e} + R\right\}.
	\]
	We are to show that $\mathcal{E}$ holds with small probability. 
In order to estimate the probability of $\mathcal{E}$, we use a peeling argument.  Let $\alpha = 12$ and $v > 0$ to be specified.  For $l \in \mathbb{N}^*$, set
	\[
		\mathcal{S}_l = \left\{A \in \mathcal{C}(r, e-s): \, \alpha^{l-1} v \leq \|A\|^2_{L_2(\Pi), s:e} \leq \alpha^l v\right\}.
	\]	
	If the event $\mathcal{E}$ holds for some matrix $A \in \mathcal{C}(r, e-s)$, then there exists $l \in \mathbb{N}$ such that $A \in \mathcal{S}_l$ and 
	\begin{align*}
		\left|\frac{1}{e-s}\sum_{t = s+1}^e \|A_{\Omega(t)}\|_{\mathrm{F}}^2 - \|A\|^2_{L_2(\Pi), s:e}\right| > \frac{1}{2} \|A\|^2_{L_2(\Pi), s:e} + R \geq \frac{1}{2} \alpha^{l-1} v + R = \frac{1}{24} \alpha^l v + R.
	\end{align*}
	For $W > v$, consider the following set
	\[
		\mathcal{C}(r, e-s, W) = \left\{A \in \mathcal{C}(r, e-s): \, \|A\|^2_{L_2(\Pi), s:e} \leq W\right\}
	\]
	and the following event
	\[
		\mathcal{E}_l = \left\{\exists A \in \mathcal{C}(r, e-s, \alpha^l v):\, \left|\frac{1}{e-s}\sum_{t = s+1}^e \|A_{\Omega(t)}\|_{\mathrm{F}}^2 - \|A\|^2_{L_2(\Pi), s:e}\right| > \frac{1}{24} \alpha^l v + R\right\}.
	\]
	Note that $A \in \mathcal{S}_l$ implies that $A \in \mathcal{C}(r, e-s, \alpha^l v)$, and consequently that $\mathcal{E} \subset \cup_{l \in \mathbb{N}^*} \mathcal{E}_l$.  This means that it suffices to estimate the probability of $\mathcal{E}_l$ and the apply a union bound argument.  
	
Let 
	\[
		Z_{W, s:e} = \sup_{A \in \mathcal{C}(r, e-s, W)} \left|\frac{1}{e-s}\sum_{t = s+1}^e \|A_{\Omega(t)}\|_{\mathrm{F}}^2 - \|A\|^2_{L_2(\Pi), s:e}\right|.
	\]
	Due to \Cref{lem-lem10-klopp_dep}, we have that 
	\[
		\mathbb{P}(\mathcal{E}_l) \leq 2\exp\left(-C_1q_1 (e-s) \alpha^l v\right).  
	\]
	Using a union bound argument, we have that
	\begin{align*}
		\mathbb{P}(\mathcal{E}) & \leq \sum_{l = 1}^{\infty} \mathbb{P}(\mathcal{E}_l) \leq 2\sum_{l = 1}^{\infty}\exp\{-C_1q_1 (e-s) \alpha^l v\} \leq 2\sum_{l = 1}^{\infty}\exp\{-C_1q_1(e-s) \log(\alpha) l v\} \\
		& \leq \frac{2\exp\{-C_1q_1(e-s) \log(\alpha) v\}}{1 - \exp\{-C_1q_1(e-s) \log(\alpha)v\}}.
	\end{align*}
	Set
	\[
		v = \frac{\log(3/\delta)}{C_1q_1(e-s) \log(\alpha)}.
	\]
	Since $\delta \in (0,1)$, we have that
	\[
		\mathbb{P}(\mathcal{E}) \leq \delta.
	\]
\end{proof}

\begin{lemma}\label{lem-lem10-klopp_dep}
Suppose the sequence of missingness matrices $\{\Omega(t)\}_{t \in \mathbb{N}} \subset \mathbb{R}^{n \times n}$ and the sequence of missingness probability matrices $\{\Pi(t)\}_{t \in \mathbb{N}^*} \subset \mathbb{R}^{n \times n}$ satisfy Assumption~\ref{assume-model-new}. For any integer pair $1 \leq s < e$ and any $W > 0$, let 
	\[
		Z_{W, s:e} = \sup_{A \in \mathcal{C}(r, e-s, W)} \left|\frac{1}{e-s}\sum_{t = s+1}^e \|A_{\Omega(t)}\|_{\mathrm{F}}^2 - \|A\|^2_{L_2(\Pi), s:e}\right|
	\]
	and
	\[
		\mathcal{C}(r, e-s, W) = \left\{A \in \mathcal{C}(r, e-s): \, \|A\|^2_{L_2(\Pi), s:e} \leq W\right\}.
	\]
	We have that
\begin{align*}
		&\mathbb{P}\left(Z_{W, e:s} \geq \frac{W}{24} + \frac{54r}{q_1}\left\{\mathbb{E}\left[\|\Sigma_{R, s:e}\|_{\mathrm{op}}\right]\right\}^2\right)
  \leq 2 \exp\left(-C q_1(e-s)W\right).
	\end{align*}
	where $C > 0$ is some absolute constant.
\end{lemma}

\begin{proof}[Proof of \Cref{lem-lem10-klopp_dep}]
\
\\
\noindent\textbf{Step 1.} We will start by showing that $Z_{W, s:e}$ concentrates around its expectation and then upper bound the conditional expectation.  Recall that
	\begin{align*}
		Z_{W, s:e} & = \sup_{A \in \mathcal{C}(r, e-s, W)} \left|\frac{1}{e-s}\sum_{t = s+1}^e \|A_{\Omega(t)}\|_{\mathrm{F}}^2 - \|A\|^2_{L_2(\Pi), s:e}\right| \\
		& = \sup_{A \in \mathcal{C}(r, e-s, W)} \left|\frac{1}{e-s} \sum_{t = s+1}^e \sum_{(i, j)} \Omega_{ij}(t)A^2_{ij} - \frac{1}{e-s} \sum_{t = s+1}^e \sum_{(i, j)} \Pi_{ij}(t) A_{ij}^2\right| \\
		& = \sup_{A \in \mathcal{C}(r, e-s, W)} \left|\frac{1}{e-s} \sum_{t = s+1}^e \sum_{(i, j)} \left\{\Omega_{ij}(t) - \Pi_{ij}(t)\right\}A^2_{ij} \right|\\
        & \leq \sup_{A \in \mathcal{C}(r, e-s, W)} \left|\frac{1}{e-s} \sum_{t = s+1}^e \sum_{i > j} \left\{\Omega_{ij}(t) - \Pi_{ij}(t)\right\}A^2_{ij} \right|\\
        & \quad + \sup_{A \in \mathcal{C}(r, e-s, W)} \left|\frac{1}{e-s} \sum_{t = s+1}^e \sum_{i \leq j} \left\{\Omega_{ij}(t) - \Pi_{ij}(t)\right\}A^2_{ij} \right|.
	\end{align*}	
We are to apply \Cref{thm-talagrand_dep}. Note that for each $t$, $\Omega(t)$ is symmetric and with upper diagonal entries being independent. It follows from \Cref{assume-model-new}c., across $t$, $\{\Omega_{ij}(t)\}_{t \in \mathbb{N}^*}$ is $\phi$-mixing with coefficients satisfy that $\phi_{\Omega}^{ij}(\ell) \leq 2\rho^{\ell}$ for any $\ell \in \mathbb N$ and some $\rho \in [0,1)$.. Thus, the sequence $\{\Omega_{ij}(t)\}_{1 \leq i \leq j \leq n, t \in \mathbb{N}^*}$ is also  $\phi$-mixing.   Let
	\[
		f(x_{ij,t};\; 1 \leq i \leq j  \leq n;\; t = s+1, \ldots, e) = \frac{1}{e-s}\sup_{A \in \mathcal{C}(r, e-s, W)} \left|\sum_{t = s+1}^e \sum_{i \leq j} \{x_{ij, t} - \Pi_{ij}(t)\} A_{ij}^2\right|.
	\]
	We first show that $f(\cdot)$ is a Lipschitz function with the constant $\sqrt{W/\{(e-s)q_1\}}$. It follows that
	\begin{align*}
		& |f(x_{ij, t}) - f(z_{ij, t})| \\
		= & \frac{1}{e-s} \left|\sup_{A \in \mathcal{C}(r, e-s, W)} \left|\sum_{t = s+1}^e \sum_{i \leq j} \{x_{ij, t} - \Pi_{ij}(t)\} A_{ij}^2\right| - \sup_{A \in \mathcal{C}(r, e-s, W)} \left|\sum_{t = s+1}^e \sum_{i \leq j} \{z_{ij, t} - \Pi_{ij}(t)\} A_{ij}^2\right|\right| \\
		\leq & \frac{1}{e-s} \sup_{A \in \mathcal{C}(r, e-s, W)} \left|\left|\sum_{t = s+1}^e \sum_{i \leq j} \{x_{ij, t} - \Pi_{ij}(t)\} A_{ij}^2\right| - \left|\sum_{t = s+1}^e \sum_{i \leq j} \{z_{ij, t} - \Pi_{ij}(t)\} A_{ij}^2\right|\right| \\
		\leq & \frac{1}{e-s} \sup_{A \in \mathcal{C}(r, e-s, W)} \left|\sum_{t = s+1}^e \sum_{i \leq j}  \{x_{ij, t} - z_{ij}(t)\} A_{ij}^2\right| \\
		\leq & \frac{1}{\sqrt{e-s}} \sup_{A \in \mathcal{C}(r, e-s, W)} \sqrt{\frac{1}{e-s}\sum_{t = s+1}^e \sum_{i \leq j} \Pi_{ij}(t) A_{ij}^4} \sqrt{\sum_{t = s+1}^e \sum_{i \leq j} (\Pi_{ij}(t))^{-1} (x_{ij, t} - z_{ij, t})^2} \\
		\leq & \frac{1}{\sqrt{(e-s)q_1}}\sup_{A \in \mathcal{C}(r, e-s, W)} \sqrt{\frac{1}{e-s}\sum_{t = s+1}^e \sum_{i \leq j} \Pi_{ij}(t) A_{ij}^2} \sqrt{\sum_{t = s+1}^e \sum_{i \leq j} (x_{ij, t} - z_{ij, t})^2}  \\
		\leq & \sqrt{\frac{W}{(e-s)q_1}} \sqrt{\sum_{t = s+1}^e \sum_{i \leq j} (x_{ij, t} - z_{ij, t})^2}.
	\end{align*}
    Similar argument applies to the lower diagonal entries, i.e.~$i > j$.
    Then \Cref{thm-talagrand_dep} leads to that
	\[
		\mathbb{P}\left(Z_{W, e:s} \geq \mathbb{E}[Z_{W, e:s}] + t\right) \leq 2\exp\left(-\frac{t^2 (e-s)q_1}{2\|\Gamma\|^2_{\mathrm{op}}W}\right),
	\]	
 where $\|\Gamma\|_{\mathrm{op}} < \infty$ under Assumption~\ref{assume-model-new}c.
	Taking $t = W/27$, we have that
	\[
		\mathbb{P}\left(Z_{W, e:s} \geq \mathbb{E}[Z_{W, e:s}] + W/27\right) \leq 2 \exp(-C q_1(e-s)W),
	\]
	with some absolute constant $C > 0$.

\medskip	
\noindent\textbf{Step 2.}  Next, we bound the expectation $\mathbb{E}[Z_{W,s:e}]$. It follows that
    \begin{align*}
		\mathbb{E}\left[Z_{W, s:e}\right] & = \frac{1}{e-s}\mathbb{E}\left[\sup_{A \in \mathcal{C}(r, s:e, W)} \left|\sum_{t = s+1}^e \sum_{(i, j)}\left\{\Omega_{ij}(t) - \Pi_{ij}(t)\right\} A^2_{ij}\right|\right]	\\
        &= \mathbb{E}\left[\sup_{A \in \mathcal{C}(r, s:e, W)} \left|\langle \Sigma_{R,s:e}, A\rangle\right|\right],
	\end{align*}
    where the second identity follows from the definition of $\Sigma_{R,s:e}$ in \eqref{eq-Sigma-prime-maat-def}.
	For $A \in \mathcal{C}(r, e-s, W)$, we have that 
	\[
		\|A\|_* \leq \sqrt{r} \|A\|_{\mathrm{F}} \leq \sqrt{\frac{r}{q_1}} \|A\|_{L_2(Q), s:e} \leq \sqrt{\frac{rW}{q_1}}.
	\]
	Therefore, we have that 
	\[
		\mathbb{E}\left[Z_{W, e:s}\right] \leq \sqrt{\frac{rW}{q_1}}\mathbb{E}\left[\|\Sigma_{R, s:e}\|_{\mathrm{op}}\right].
	\]

\medskip	
\noindent\textbf{Step 3.}  Finally, using
    \begin{align*}
        \mathbb{E}\left[Z_{W, e:s}\right] + \frac{W}{27} \leq& \frac{54r}{q_1}\left(\mathbb{E}\left[\|\Sigma_{R, s:e}\|_{\mathrm{op}}\right]\right)^2 + \frac{W}{54\times 4} + \frac{W}{27}\\
        =& \frac{54r}{q_1}\left(\mathbb{E}\left[\|\Sigma_{R, s:e}\|_{\mathrm{op}}\right]\right)^2 +\frac{W}{24},
    \end{align*}
we have that
	\[
		\mathbb{P}\left(Z_{W, e:s} \geq \frac{W}{24} + \frac{54r}{q_1}\left(\mathbb{E}\left[\|\Sigma_{R, s:e}\|_{\mathrm{op}}\right]\right)^2\right) \leq 2 \exp(-C q_1(e-s)W).
	\]
\end{proof}

\begin{lemma}\label{lem-lem4-klopp_dep}
Suppose Assumption~\ref{assume-model-new} holds. For any integer pair $1 \leq s < e$, such that the interval length $(e-s)$ satisfying $$\frac{\sqrt{e-s}}{(\log(e-s))^2} \geq C\sqrt{n\log(n)}.$$ Then\\
(i) we have that
    \begin{align*}
    \mathbb{E}\left[\| \Sigma_{R,s:e}\|_{\mathrm{op}} \right]\leq C_1\sqrt{\frac{n\log(n)}{e-s}};
\end{align*} 
\\
(ii) we also have with probability at least $1 - \delta$ that
\begin{align*} \|\Sigma_{s:e}\|_{\mathrm{op}} \leq C_2\sqrt{\frac{n\{\log(n) + \log(1/\delta)\}}{e-s}} + C_3\sqrt{\vartheta}q_2n\sqrt{\frac{\log(n) + \log(1/\delta)}{e-s}},
 \end{align*}
\end{lemma}
\begin{proof}
(i) \textbf{Upper bound of $\mathbb{E}[\|\Sigma_{R, s:e}\|_{\mathrm{op}}]$}.
\\
Recall from \eqref{eq-Sigma-prime-maat-def} that
    \begin{align*}
        \Sigma_{R,s:e} =\frac{1}{e-s}\sum_{(i,j)}E_{ij}\sum_{t = s+1}^e\big\{\Omega_{ij}(t) - \Pi_{ij}(t)\big\}.
    \end{align*}
It follows from \Cref{assume-model-new} that, for each $t$, the upper diagonal entries of $\Omega(t)$ are independent. To apply \Cref{lemma:matrix_bernstein}, we verify its conditions. Note that
\[
\mathbb{E}[\Omega(t) - \Pi(t)] = 0,
\]
\[
\max_{t}\|\Omega(t) - \Pi(t)\|_{\op} \leq \max_{t}\|\Omega(t) - \Pi(t)\|_{\mathrm{F}} \leq n = D,
\]
and
\begin{align*}
    &\nu^2 = \max_{1 \leq s < e}\frac{1}{e-s}\left\|\mathbb{E}\left[\left(\sum_{t = s+1}^e\{\Omega(t) - \Pi(t)\}\right)\left(\sum_{t = s+1}^e\{\Omega(t) - \Pi(t)\}\right)^{\top}\right]\right\|_{\mathrm{op}}\\
    =&\max_{1 \leq s < e}\frac{1}{e-s}\sup_{\|v\| = 1}\mathbb{E}\left[\sum_{p = 1}^n\sum_{q=1}^n v_p\left\{\left(\sum_{t = s+1}^e\{\Omega(t) - \Pi(t)\}\right)\left(\sum_{t = s+1}^e\{\Omega(t) - \Pi(t)\}\right)^{\top}\right\}_{p,q}v_q\right]\\
    =&\max_{1 \leq s < e}\frac{1}{e-s}\sup_{\|v\| = 1}\mathbb{E}\Bigg[\sum_{p = 1}^n\sum_{q=1}^n v_p\Bigg\{\sum_{k=1}^n\left(\sum_{t = s+1}^e\{\Omega_{pk}(t) - \Pi_{pk}(t)\}\right)\left(\sum_{t = s+1}^e\{\Omega_{qk}(t) - \Pi_{qk}(t)\}\right)\Bigg\}v_q\Bigg]\\
        =&\max_{1 \leq s < e}\frac{1}{e-s}\sup_{\|v\| = 1}\mathbb{E}\Bigg[\sum_{k=1}^n\sum_{p = 1}^n v_p^2\left(\sum_{t = s+1}^e\{\Omega_{pk}(t) - \Pi_{pk}(t)\}\right)^2\Bigg]\\
        =&\max_{1 \leq s < e}\frac{1}{e-s}\sup_{\|v\| = 1}\sum_{k=1}^n\sum_{p = 1}^n v_p^2\mathbb{E}\Bigg[\left(\sum_{t = s+1}^e\{\Omega_{pk}(t) - \Pi_{pk}(t)\}\right)^2\Bigg]\\
        \leq&\max_{1 \leq s < e}\frac{n}{e-s}\sup_{\|v\| = 1}\sum_{p = 1}^n v_p^2\max_{p,k}\mathbb{E}\Bigg[\left(\sum_{t = s+1}^e\{\Omega_{pk}(t) - \Pi_{pk}(t)\}\right)^2\Bigg]\\
        \leq&\max_{1 \leq s < e}\frac{n}{e-s}\max_{p,k}\mathbb{E}\Bigg[\left(\sum_{t = s+1}^e\{\Omega_{pk}(t) - \Pi_{pk}(t)\}\right)^2\Bigg]\\
    \leq& n\max_{p,k}\sum_{\ell = -\infty}^{\infty}\cov\left(\Omega_{pk}(s) - \Pi_{pk}(s), \Omega_{pk}(s+\ell) - \Pi_{pk}(s+\ell)\right)\\
    =& C_{\Omega,\mathrm{LRV}}n,
\end{align*}
where the fourth inequality follows from
\[\mathbb{E}\left[\left(\sum_{t = s+1}^e \left\{\Omega_{pk}(t) - \Pi_{pk}(t)\right\}\right)\left(\sum_{t = s+1}^e \left\{\Omega_{qk}(t) - \Pi_{qk}(t)\right\}\right)\right] = 0,\]
for any $p \neq q$. The last equality follows from \Cref{assume-model-new}d, which implies for any $p,k = 1, \dots, n$ the long-run variance of $\{\Omega_{pk}(t)\}_{t \in \mathbb{N}^*}$ is finite, i.e. $$C_{\Omega,\mathrm{LRV}} = \max_{p,k}\sum_{\ell = -\infty}^{\infty}\cov\left(\Omega_{pk}(s) - \Pi_{pk}(s), \Omega_{pk}(s+\ell) - \Pi_{pk}(s+\ell)\right) < \infty.$$
Then, \Cref{lemma:matrix_bernstein} leads to
\begin{align*}
    &\mathbb{P}\left(\sqrt{\frac{e-s}{n}}\left\|\Sigma_{R,s:e}\right\|_{\mathrm{op}} \geq x\right)\\
    \leq& n\exp\left(-\frac{C(e-s)nx^2}{C_{\Omega,\mathrm{LRV}}n(e-s) + c^{-1}n + x\sqrt{n(e-s)}n(\log(e-s))^2}\right)\\
    \leq& n\exp\left(-C_1x^2\right) + n\exp\left(-C_2(e-s)x^2\right) + n\exp\left(-\frac{C_3\sqrt{e-s}x}{\sqrt{n}(\log(e-s))^2}\right),
\end{align*}
which further leads to
\begin{align*}
    &\mathbb{P}\left(\left\|\Sigma_{R,s:e}\right\|_{\mathrm{op}} \geq C\sqrt{\frac{n\log(n)}{e-s}}\right)
    \leq n^{-5},
\end{align*}
provided that the interval length $(e-s)$ satisfying $$\frac{\sqrt{e-s}}{(\log(e-s))^2} \geq \sqrt{n\log(n)}.$$
Since 
    \[
    \left\|\Sigma_{R,s:e}\right\|_{\mathrm{op}} \leq \frac{n}{e-s}\left\|\sum_{(i,j)}E_{ij}\sum_{t = s+1}^e\{\Omega_{ij}(t) - \Pi_{ij}(t)\}\right\|_{\max} \leq n.
    \]
    Taking integral of the tail probability bound leads to
\begin{align*}
    &\mathbb{E}\left[\left\|\Sigma_{R,s:e}\right\|_{\mathrm{op}}\right]
    \leq C_1 \sqrt{\frac{n\log(n)}{e-s}}(1-n^{-5}) + C_2n^{-4} \leq C\sqrt{\frac{n\log(n)}{e-s}},
\end{align*}
\\
\\
(ii) \textbf{Tail probability bound of $\|\Sigma_{s:e}\|_{\mathrm{op}}$}.
\\
Recall the definition \eqref{eq-Sigma-mat-def} that
	\begin{align*}
	    \Sigma_{s:e} = \frac{1}{e-s} \sum_{(i, j)}E_{ij} \sum_{t = s+1}^e \left(\Omega_{ij}(t)\big\{P_{ij}(t) + \xi_{ij}(t) - \Theta_{ij}(t)\big\}\right).
	\end{align*}
For any $t > 0$, it holds that
\begin{align*}
     &\mathbb P\left(\|\Sigma_{s:e}\|_{\mathrm{op}} \geq t \right)\\ \leq&
     \mathbb P\left(\left\|\frac{1}{e-s} \sum_{(i, j)}E_{ij} \sum_{t = s+1}^e \left(\Omega_{ij}(t) \{P_{ij}(t) - \Theta_{ij}(t)\}\right)\right\|_{\mathrm{op}} \geq \frac{2t}{3} \right)\\
     &+ \mathbb P\left(\left\|\frac{1}{e-s} \sum_{(i, j)}E_{ij} \sum_{t = s+1}^e \Omega_{ij}(t)\xi_{ij}(t)\right\|_{\mathrm{op}} \geq \frac{t}{3} \right)\\
     =&P\left(\left\|\frac{1}{e-s} \sum_{(i, j)}E_{ij} \sum_{t = s+1}^e \left(\{\Omega_{ij}(t) - \Pi_{ij}(t)\} \{P_{ij}(t) - \Theta_{ij}(t)\}\right)\right\|_{\mathrm{op}} \geq \frac{t}{3}\right)\\
     &+ P\left(\left\|\frac{1}{e-s} \sum_{(i, j)}E_{ij} \sum_{t = s+1}^e \left(\Pi_{ij}(t) \{P_{ij}(t) - \Theta_{ij}(t)\}\right)\right\|_{\mathrm{op}} \geq \frac{t}{3}\right)\\
     &+ \mathbb E\left[\mathbb P\left(\left\|\frac{1}{e-s} \sum_{(i, j)}E_{ij} \sum_{t = s+1}^e \Omega_{ij}(t)\xi_{ij}(t)\right\|_{\mathrm{op}} \geq \frac{t}{3} \Bigg| \{X_i(t)\}_{i = 1,t\in\mathbb{N}^*}^n\right)\right]\\
     =&I + II + III.
 \end{align*}
\\
\\
\textbf{Term $I$}.
\
\\
Under Assumption~\ref{assume-model-new}, $\{\Omega(t)\}_{t = s+1}^e$ and $\{P(t)\}_{t = s+1}^e$ are sequences of symmetric matrices and are mutually independent. Moreover, for each $t$, the upper diagonal entries of $\Omega(t)$ are independent across $i,j$. 
\
\\
To apply \Cref{lemma:matrix_bernstein}, we verify its conditions. For any $t = s+1, \dots, e$, let $M(t) = \{\Omega(t) - \Pi(t)\} \odot \{P(t) - \Theta(t)\}$. Note that
\[
\mathbb{E}[M(t)] = 0,
\]
\[
\max_{t}\|M(t)\|_{\op} \leq \max_{t}\|M(t)\|_{\mathrm{F}} \leq n\vartheta,
\]
and following the same arguments as in the part (i),
\begin{align*}
    &\nu^2 = \max_{1 \leq s < e}\frac{1}{e-s}\left\|\mathbb{E}\left[\left(\sum_{t = s+1}^eM(t)\right)\left(\sum_{t = s+1}^eM(t)\right)^{\top}\right]\right\|_{\mathrm{op}}\\
    =&\max_{1 \leq s < e}\frac{1}{e-s}\sup_{\|v\| = 1}\mathbb{E}\left[\sum_{p = 1}^n\sum_{q=1}^n v_p\left\{\left(\sum_{t = s+1}^eM(t)\right)\left(\sum_{t = s+1}^eM(t)\right)^{\top}\right\}_{p,q}v_q\right]\\
    =&\max_{1 \leq s < e}\frac{1}{e-s}\sup_{\|v\| = 1}\mathbb{E}\Bigg[\sum_{k=1}^n\sum_{p = 1}^n\sum_{q=1}^n v_p\Bigg\{\left(\sum_{t = s+1}^e\{\Omega_{pk}(t)-\Pi_{pk}(t)\}\{(X_p(t))^{\top}X_k(t)-\Theta_{pk}(t)\}\right)\\
    &\quad\quad\quad\quad\quad\quad\quad\quad\quad\quad\quad\quad\quad\left(\sum_{t = s+1}^e\{\Omega_{qk}(t) - \Pi_{qk}(t)\}\{(X_q(t))^{\top}X_k(t)-\Theta_{qk}(t)\}\right)\Bigg\}v_q\Bigg]\\
    =&\max_{1 \leq s < e}\frac{1}{e-s}\sup_{\|v\| = 1}\mathbb{E}\Bigg[\sum_{k=1}^n\sum_{p = 1}^n v_p^2\left(\sum_{t = s+1}^e\{\Omega_{pk}(t)-\Pi_{pk}(t)\}\{(X_p(t))^{\top}X_k(t)-\Theta_{pk}(t)\}\right)^2\Bigg]\\
    \leq&\max_{1 \leq s < e}\frac{n}{e-s}\max_{p,k}\mathbb{E}\left[\left(\sum_{t = s+1}^e\{\Omega_{pk}(t)-\Pi_{pk}(t)\}\{(X_p(t))^{\top}X_k(t)-\Theta_{pk}(t)\}\right)^2\right]\\
    =&n\max_{1 \leq s < e}\max_{p,k}\sum_{\ell = -(e-s)}^{e-s}\cov\Big(\{\Omega_{pk}(s)-\Pi_{pk}(s)\}\{(X_p(s))^{\top}X_k(s)-\Theta_{pk}(s)\},\\
    &\quad\quad\quad\quad\quad\quad\quad\{\Omega_{pk}(s+\ell)-\Pi_{pk}(s+\ell)\}\{(X_p(s+\ell))^{\top}X_k(s+\ell)-\Theta_{pk}(s+\ell)\}\Big)\\
    \leq& C_{\Omega P,\mathrm{LRV}}n,
\end{align*}
where the last equality follows from \Cref{assume-model-new}b and d, which imply for any $p,k = 1, \dots, n$ the long-run variance of $\big\{\{\Omega_{pk}(t)-\Pi_{pk}(t)\}\{(X_{p}(t))^{\top}X_k(t) -\Theta_{pk}(t)\}\big\}_{t \in \mathbb{N}^*}$ is finite, i.e. 
\begin{align*}
    &C_{\Omega P,\mathrm{LRV}} = \max_{p,k}\sum_{\ell = -(e-s)}^{e-s}\cov\Big(\{\Omega_{pk}(s)-\Pi_{pk}(s)\}\{(X_p(s))^{\top}X_k(s)-\Theta_{pk}(s)\},\\
    &\quad\quad\quad\quad\quad\quad\quad\{\Omega_{pk}(s+\ell)-\Pi_{pk}(s+\ell)\}\{(X_p(s+\ell))^{\top}X_k(s+\ell)-\Theta_{pk}(s+\ell)\}\Big) < \infty.
\end{align*}
Then, \Cref{lemma:matrix_bernstein} leads to
\begin{align*}
    &P\left(\frac{1}{\sqrt{(e-s)n\vartheta }} \left\|\sum_{(i, j)}E_{ij} \sum_{t = s+1}^e \left(\{\Omega_{ij}(t) - \Pi_{ij}(t)\} \{P_{ij}(t) - \Theta_{ij}(t)\}\right)\right\|_{\mathrm{op}} \geq \frac{t}{3}\right)\\
    \leq& n\exp\left(-\frac{C(e-s)n\vartheta t^2}{C_{\Omega P,\mathrm{LRV}}n(e-s) + c^{-1}n\vartheta + t\sqrt{n\vartheta(e-s)}n\vartheta(\log(e-s))^2}\right)\\
    \leq& n\exp\left(-C_1\vartheta t^2\right) + n\exp\left(-C_2(e-s)t^2\right) + n\exp\left(-\frac{C_3\sqrt{e-s}t}{\sqrt{n\vartheta }(\log(e-s))^2}\right),
\end{align*}
which further leads to
\begin{align*}
    &\mathbb{P}\left(\frac{1}{e-s}\left\|\sum_{(i, j)}E_{ij} \sum_{t = s+1}^e \left(\{\Omega_{ij}(t) - \Pi_{ij}(t)\} \{P_{ij}(t) - \Theta_{ij}(t)\}\right)\right\|_{\mathrm{op}} \geq C_1\sqrt{\frac{n\{\log(n)+\log(1/\delta)\}}{e-s}}\right)
    \leq \delta/3,
\end{align*}
provided that the interval length $(e-s)$ satisfying $$\frac{\sqrt{e-s}}{(\log(e-s))^2} \geq C_2\sqrt{n\log(n)}\vartheta .$$
\\
\\
\textbf{Term $II$}.
To apply \Cref{lemma:matrix_bernstein}, we verify its conditions. For any $t = s+1, \dots, e$, let $M(t) = \Pi(t) \odot \{P(t) - \Theta(t)\}$. Note that
\[
\mathbb{E}[M(t)] = 0,
\]
\[
\max_{t}\|M(t)\|_{\op} \leq \max_{t}\|M(t)\|_{\mathrm{F}} \leq n\vartheta ,
\]
and following the same arguments as in the part (i),
\begin{align*}
    &\nu^2 = \max_{1 \leq s < e}\frac{1}{e-s}\left\|\mathbb{E}\left[\left(\sum_{t = s+1}^eM(t)\right)\left(\sum_{t = s+1}^eM(t)\right)^{\top}\right]\right\|_{\mathrm{op}}\\
    =&\max_{1 \leq s < e}\frac{1}{e-s}\sup_{\|v\| = 1}\mathbb{E}\left[\sum_{p = 1}^n\sum_{q=1}^n v_p\left\{\left(\sum_{t = s+1}^eM(t)\right)\left(\sum_{t = s+1}^eM(t)\right)^{\top}\right\}_{p,q}v_q\right]\\
    =&\max_{1 \leq s < e}\frac{1}{e-s}\sup_{\|v\| = 1}\mathbb{E}\Bigg[\sum_{k=1}^n\sum_{p = 1}^n\sum_{q=1}^n v_p\Bigg\{\left(\sum_{t = s+1}^e\Pi_{pk}(t)\{(X_p(t))^{\top}X_k(t)-\Theta_{pk}(t)\}\right)\\
    &\quad\quad\quad\quad\quad\quad\quad\quad\quad\quad\quad\quad\quad\left(\sum_{t = s+1}^e\Pi_{qk}(t)\{(X_q(t))^{\top}X_k(t)-\Theta_{qk}(t)\}\right)\Bigg\}v_q\Bigg]\\
    \leq&\max_{1 \leq s < e}\frac{q_2^2}{e-s}\sup_{\|v\| = 1}\mathbb{E}\Bigg[\sum_{k=1}^n\sum_{p = 1}^n\sum_{q=1}^n v_p\Bigg\{\left(\sum_{t = s+1}^e\{(X_p(t))^{\top}X_k(t)-\Theta_{pk}(t)\}\right)\\
    &\quad\quad\quad\quad\quad\quad\quad\quad\quad\quad\quad\quad\quad\left(\sum_{t = s+1}^e\{(X_q(t))^{\top}X_k(t)-\Theta_{qk}(t)\}\right)\Bigg\}v_q\Bigg]\\
    =&\max_{1 \leq s < e}\frac{q_2^2}{e-s}\sup_{\|v\| = 1}\mathbb{E}\Bigg[\sum_{k=1}^n\sum_{p = 1}^n v_p^2\left(\sum_{t = s+1}^e\{(X_p(t))^{\top}X_k(t)-\Theta_{pk}(t)\}\right)^2\Bigg]\\
    &+\max_{1 \leq s < e}\frac{q_2^2\vartheta }{e-s}\sup_{\|v\| = 1}\mathbb{E}\Bigg[\sum_{k=1}^n\sum_{p = 1}^n\sum_{q \neq p} v_p\left(\sum_{t = s+1}^e(X_k(t) - \mu_k(t))\right)^{\top}\left(\sum_{t = s+1}^e(X_k(t) - \mu_k(t))\right)v_q\Bigg]\\
    =&\max_{1 \leq s < e}\frac{q_2^2}{e-s}\sup_{\|v\| = 1}\mathbb{E}\Bigg[\sum_{k=1}^n\sum_{p = 1}^n v_p^2\left(\sum_{t = s+1}^e\{(X_p(t))^{\top}X_k(t)-\Theta_{pk}(t)\}\right)^2\Bigg]\\
    &+\max_{1 \leq s < e}\frac{q_2^2\vartheta }{e-s}\sum_{k=1}^n\mathbb{E}\Bigg[\left\|\sum_{t = s+1}^e(X_k(t) - \mu_k(t))\right\|_2^2\Bigg]\sup_{\|v\| = 1}\sum_{p = 1}^n\sum_{q \neq p} v_pv_q\\
    \leq&\max_{1 \leq s < e}\frac{q_2^2n}{e-s}\max_{p,k}\mathbb{E}\Bigg[\left(\sum_{t = s+1}^e\{(X_p(t))^{\top}X_k(t)-\Theta_{pk}(t)\}\right)^2\Bigg]\\
    &+\max_{1 \leq s < e}\frac{q_2^2\vartheta n(n-1)}{e-s}\max_{k}\mathbb{E}\Bigg[\left\|\sum_{t = s+1}^e(X_k(t) - \mu_k(t))\right\|_2^2\Bigg]\\
    \leq& C_{ P,\mathrm{LRV}}q_2^2n + C_{X,\mathrm{LRV}}q_2^2\vartheta n(n-1),
\end{align*}
where the fourth equality follows from for any $p \neq q$ and $t_1, t_2$,
\begin{align*}
    &\mathbb{E}\left[\{(X_p(t_1))^{\top}X_k(t_1)-\Theta_{pk}(t_1)\}\{(X_q(t_2))^{\top}X_k(t_2)-\Theta_{qk}(t_2)\}\right]\\
    =&\mathbb{E}\left[\{(X_p(t_1))^{\top}X_k(t_1)-(\mu_{p}(t_1))^{\top}\mu_{k}(t_1)\}\{(X_q(t_2))^{\top}X_k(t_2)-(\mu_{q}(t_2))^{\top}\mu_k(t_2)\}\right]\\
    =&\mathbb{E}\Big[\{(X_p(t_1) - \mu_p(t_1))^{\top}(X_k(t_1)-\mu_k(t_1)) + (X_k(t_1) - \mu_k(t_1))^{\top}\mu_p(t_1) + (X_p(t_1) - \mu_p(t_1))^{\top}\mu_k(t_1)\}\\
    &\quad\quad\{(X_q(t_2) - \mu_q(t_2))^{\top}(X_k(t_2)-\mu_k(t_2)) + (X_k(t_2) - \mu_k(t_2))^{\top}\mu_q(t_2) + (X_q(t_2) - \mu_q(t_2))^{\top}\mu_k(t_2)\}\Big]\\
    =&\mathbb{E}\Big[(X_p(t_1) - \mu_p(t_1))^{\top}(X_k(t_1)-\mu_k(t_1))(X_q(t_2) - \mu_q(t_2))^{\top}(X_k(t_2)-\mu_k(t_2))\Big]\\
    &+
    \mathbb{E}\big[(X_k(t_1) - \mu_k(t_1))^{\top}\mu_p(t_1)(X_k(t_2) - \mu_k(t_2))^{\top}\mu_q(t_2)\big]\\
    =& \mathbb{E}\big[(X_k(t_1) - \mu_k(t_1))^{\top}\mu_p(t_1)(X_k(t_2) - \mu_k(t_2))^{\top}\mu_q(t_2)\big]\\
    \leq& \vartheta \mathbb{E}\big[(X_k(t_1) - \mu_k(t_1))^{\top}(X_k(t_2) - \mu_k(t_2))\big].
\end{align*}
Moreover, the last inequality follows from \Cref{assume-model-new}b, which implies for any $p,k = 1, \dots, n$ the long-run variance of $\{(X_p(t))^{\top}X_k(t)\}_{t \in \mathbb{N}^*}$ is finite, i.e. $$C_{P,\mathrm{LRV}} = \max_{p,k}\sum_{\ell = -\infty}^{\infty}\cov\left((X_p(t))^{\top}X_k(t), (X_p(t+\ell))^{\top}X_k(t+\ell)\right) < \infty,$$
and for any $p = 1, \dots, n$ the long-run variance of $\{X_p(t)\}_{t \in \mathbb{N}^*}$ is finite, i.e. $$C_{X,\mathrm{LRV}} = \max_{p,k}\sum_{\ell = -\infty}^{\infty}\mathbb{E}\left[(X_p(s) - \mu_p(s))^{\top}(X_p(s+\ell) - \mu_p(s+\ell))\right] < \infty,$$
and $\mu_p(t) = \mathbb{E}[X_p(t)] \in \mathcal{X}_F$ defined in \Cref{def-inner-product-dist}.
Then, \Cref{lemma:matrix_bernstein} leads to
\begin{align*}
    &P\left(\frac{1}{\sqrt{(e-s)\vartheta }q_2n} \left\|\sum_{(i, j)}E_{ij} \sum_{t = s+1}^e \left(\Pi_{ij}(t)\{P_{ij}(t) - \Theta_{ij}(t)\}\right)\right\|_{\mathrm{op}} \geq \frac{t}{3}\right)\\
    \leq& n\exp\left(-\frac{C(e-s)n^2q_2^2\vartheta t^2}{C_{ P,\mathrm{LRV}}q_2^2n(e-s) + C_{X,\mathrm{LRV}}q_2^2\vartheta n(n-1)(e-s) + c^{-1}n\vartheta  + t\sqrt{\vartheta (e-s)}n^2q_2\vartheta (\log(e-s))^2}\right)\\
    \leq& n\exp\left(-C_1n\vartheta t^2\right) + n\exp\left(-C_2t^2\right) + n\exp\left(-C_3nq_2^2(e-s)t^2\right) + n\exp\left(-\frac{C_4q_2\sqrt{e-s}t}{\sqrt{\vartheta }(\log(e-s))^2}\right),
\end{align*}
which further leads to
\begin{align*}
    &\mathbb{P}\left(\frac{1}{e-s}\left\|\sum_{(i, j)}E_{ij} \sum_{t = s+1}^e \Pi_{ij}(t) \{P_{ij}(t) - \Theta_{ij}(t)\}\right\|_{\mathrm{op}} \geq C_1\sqrt{\vartheta }q_2n\sqrt{\frac{\log(n) + \log(1/\delta)}{e-s}}\right)
    \leq \delta/3,
\end{align*}
provided that the interval length $(e-s)$ satisfying $$\frac{\sqrt{e-s}}{(\log(e-s))^2} \geq C_2\frac{\sqrt{\vartheta \log(n)}}{q_2}.$$
\\
\\
\textbf{Term $III$}.
\
\\
Under Assumption~\ref{assume-model-new}, $\{\Omega(t)\}_{t = s+1}^e$ and $\{\xi(t)\}_{t = s+1}^e$ are sequences of symmetric matrices and are mutually independent. Moreover, for each $t$, the upper diagonal entries of $\Omega(t)$ are independent across $i,j$. Conditioning on the latent positions $\{X_i(t)\}_{i = 1, t \in \mathbb{N}^*}^n$, for each $t$, the upper diagonal entries of $\xi(t)$ are independent across $i,j$.
\
\\
To apply \Cref{lemma:matrix_bernstein}, we verify its conditions. For any $t = s+1, \dots, e$, let $M(t) = \Omega(t)\odot \xi(t)$. Note that
\[
\mathbb{E}[M(t)|\{X_i(t)\}_{i = 1, t \in \mathbb{N}^*}^n] = 0,
\]
\[
\max_{t}\|M(t)|\{X_i(t)\}_{i = 1, t \in \mathbb{N}^*}^n\|_{\op} \leq \max_{t}\|M(t)|\{X_i(t)\}_{i = 1, t \in \mathbb{N}^*}^n\|_{\mathrm{F}} \leq n,
\]
and following the same arguments as in the part (i),
\begin{align*}
    &\nu^2 = \max_{1 \leq s < e}\frac{1}{e-s}\left\|\mathbb{E}\left[\left(\sum_{t = s+1}^eM(t)\right)\left(\sum_{t = s+1}^eM(t)\right)^{\top}\Bigg|\{X_i(t)\}_{i = 1, t \in \mathbb{N}^*}^n\right]\right\|_{\mathrm{op}}\\
    =&\max_{1 \leq s < e}\frac{1}{e-s}\sup_{\|v\| = 1}\mathbb{E}\Bigg[\sum_{k=1}^n\sum_{p = 1}^n\sum_{q=1}^n v_p\Bigg\{\left(\sum_{t = s+1}^e\Omega_{pk}(t)\xi_{pk}(t)\right)\left(\sum_{t = s+1}^e\Omega_{qk}(t)\xi_{qk}(t)\right)\Bigg\}v_q\Bigg|\{X_i(t)\}_{i = 1, t \in \mathbb{N}^*}^n\Bigg]\\
    =&\max_{1 \leq s < e}\frac{1}{e-s}\sup_{\|v\| = 1}\mathbb{E}\Bigg[\sum_{k=1}^n\sum_{p = 1}^n v_p^2\left(\sum_{t = s+1}^e\Omega_{pk}(t)\xi_{pk}(t)\right)^2\Bigg|\{X_i(t)\}_{i = 1, t \in \mathbb{N}^*}^n\Bigg]\\
    \leq&\max_{1 \leq s < e}\frac{n}{e-s}\max_{p,k}\mathbb{E}\left[\left(\sum_{t = s+1}^e\Omega_{pk}(t)\xi_p(t)\right)^2\Bigg|\{X_i(t)\}_{i = 1, t \in \mathbb{N}^*}^n\right]\\
    \leq& C_{\Omega Y,\mathrm{LRV}}n,
\end{align*}
where the last equality follows from \Cref{assume-model-new}b and d, which imply for any $p,k = 1, \dots, n$ the long-run variance conditioning on $\{X_i(t)\}_{i = 1, t \in \mathbb{N}^*}^n$ of $\big\{\Omega_{pk}(t)\xi_{pk}(t)\big\}_{t \in \mathbb{N}^*}$ is finite, i.e. 
\begin{align*}
    &C_{\Omega Y,\mathrm{LRV}} = \max_{p,k}\sum_{\ell = -(e-s)}^{e-s}\cov\Big(\Omega_{pk}(s)\xi_{pk}(s),\Omega_{pk}(s+\ell)\xi_{pk}(s+\ell) \Big|\{X_i(t)\}_{i = 1, t \in \mathbb{N}^*}^n\Big) < \infty.
\end{align*}
Then, \Cref{lemma:matrix_bernstein} leads to
\begin{align*}
    &P\left(\frac{1}{\sqrt{(e-s)n}} \left\|\sum_{(i, j)}E_{ij} \sum_{t = s+1}^e \Omega_{ij}(t)\xi_{ij}(t)\right\|_{\mathrm{op}} \geq \frac{t}{3}\Bigg|\{X_i(t)\}_{i = 1, t \in \mathbb{N}^*}^n\right)\\
    \leq& n\exp\left(-\frac{C(e-s)nt^2}{C_{\Omega Y,\mathrm{LRV}}n(e-s) + c^{-1}n + t\sqrt{n(e-s)}n(\log(e-s))^2}\right)\\
    \leq& n\exp\left(-C_1t^2\right) + n\exp\left(-C_2(e-s)t^2\right) + n\exp\left(-\frac{C_3\sqrt{e-s}t}{\sqrt{n}(\log(e-s))^2}\right),
\end{align*}
which further leads to
\begin{align*}
    &\mathbb{P}\left(\frac{1}{e-s}\left\|\sum_{(i, j)}E_{ij} \sum_{t = s+1}^e \Omega_{ij}(t)\xi_{ij}(t)\right\|_{\mathrm{op}} \geq C_1\sqrt{\frac{n\{\log(n) + \log(1/\delta)\}}{e-s}}\right)\\
    &=\mathbb{E}\left[\mathbb{P}\left(\frac{1}{e-s}\left\|\sum_{(i, j)}E_{ij} \sum_{t = s+1}^e \Omega_{ij}(t)\xi_{ij}(t)\right\|_{\mathrm{op}} \geq C_1\sqrt{\frac{n\{\log(n) + \log(1/\delta)\}}{e-s}}\Bigg|\{X_i(t)\}_{i = 1, t \in \mathbb{N}^*}^n \right)\right]
    \leq \delta/3,
\end{align*}
provided that the interval length $(e-s)$ satisfying $$\frac{\sqrt{e-s}}{(\log(e-s))^2} \geq C_2\sqrt{n\log(n)}.$$
\\
\\
Combining the above three terms, we have with probability at least $1 - \delta$ that
\begin{align*}
\|\Sigma_{s:e}\|_{\mathrm{op}} \leq C_1\sqrt{\frac{n\{\log(n) + \log(1/\delta)\}}{e-s}} + C_2\sqrt{\vartheta }q_2n\sqrt{\frac{\log(n) + \log(1/\delta)}{e-s}},
 \end{align*}
 provided that the interval length $(e-s)$ satisfying $$\frac{\sqrt{e-s}}{(\log(e-s))^2} \geq C_3\sqrt{n\log(n)}.$$
\end{proof}

\end{document}